\documentclass[smallextended,envcountsame,envcountsect]{svjour3}
\smartqed
\journalname{Finance and Stochastics}
\def\subclassname{{\bfseries Mathematics Subject Classification}\enspace}

\usepackage{amsmath,amssymb}
\usepackage{amsfonts}
\usepackage[paperwidth=439.37pt,paperheight=666.142pt,left=52pt,right=51pt,top=54pt,bottom=61pt,headsep=16pt]{geometry}
\usepackage{etoolbox}
\makeatletter
\renewcommand\makeheadbox{}%
\patchcmd{\@maketitle}{\hrule\@height0.35mm}{}{}{\typeout{RULEPATCHFAIL}}%
\makeatother

\usepackage{graphicx}
\usepackage{xcolor}
\usepackage{lastpage}

\usepackage{booktabs}
\usepackage{enumitem}
\usepackage{listings}
\usepackage[english]{babel}
\usepackage{framed}
\usepackage{mathrsfs}
\usepackage{pifont}
\usepackage{xparse}
\usepackage{float}
\usepackage{tikz}
\usetikzlibrary{fit,shapes.geometric}

\allowdisplaybreaks

\numberwithin{equation}{section}

\definecolor{Kblue}{RGB}{0,40,104}
\definecolor{Kgrey}{RGB}{224,224,224}
\definecolor{Kgreen}{RGB}{0,153,0}

\usepackage{hyperref}

\hypersetup{
    colorlinks=true,
    linkcolor=Kblue,
    citecolor=blue,
    urlcolor=Kblue
}

\makeatletter

\def\eps{\varepsilon}

\NewDocumentCommand{\Expxt}{m g}{%
\IfNoValueTF{#2}%
 {\relax\if@display
         \mathbb{E}_{x_t}\!\left[{#1}\right]%
         \else
         \mathbb{E}_{x_t}[{#1}]%
         \fi}%
{\relax\if@display
\mathbb{E}_{x_t}\!\left[{#1}\,\middle|\,{#2}\right]%
\else
\mathbb{E}_{x_t}[{#1}\,|\,{#2}]%
\fi}%
}

\NewDocumentCommand{\Expx}{m g}{%
\IfNoValueTF{#2}%
 {\relax\if@display
         \mathbb{E}_x\!\left[{#1}\right]%
         \else
         \mathbb{E}_x[{#1}]%
         \fi}%
{\relax\if@display
\mathbb{E}_x\!\left[{#1}\,\middle|\,{#2}\right]%
\else
\mathbb{E}_x[{#1}\,|\,{#2}]%
\fi}%
}

\NewDocumentCommand{\Exp}{m g}{%
\IfNoValueTF{#2}%
 {\relax\if@display
         \mathbb{E}\!\left[{#1}\right]%
         \else
         \mathbb{E}[{#1}]%
         \fi}%
{\relax\if@display
\mathbb{E}\!\left[{#1}\,\middle|\,{#2}\right]%
\else
\mathbb{E}[{#1}\,|\,{#2}]%
\fi}%
}

\makeatother

\DeclareMathOperator{\supp}{supp}

\newcommand{\E}{{\mathbb{E}}}
\renewcommand{\P}{{\mathbb{P}}}
\newcommand{\1}{\mathbf{1}}

\begin{document}

\title{Optimal Prediction of Resistance and Support Levels under Constant Elasticity of Variance Processes}
\titlerunning{Optimal prediction of resistance and support levels under CEV}
\author{Ruibo Ma}
\authorrunning{R. Ma}
\institute{R. Ma \at Department of Mathematics, The University of Manchester, Oxford Road, Manchester M13 9PL, United Kingdom \\ \email{ruibo.ma@manchester.ac.uk}}
\date{}
\maketitle

\begin{abstract}
Assuming that the asset price $X$ follows a constant elasticity of variance process, we study the optimal prediction problem $\inf_{0\leq \tau\leq T}\E|X_\tau-\ell|$ over stopping times of $X$, where $\ell$ is a hidden aspiration level independent of $X$. Under the aspiration level hypothesis, we show that for a class of admissible laws the boundaries are located relative to the median interval of $\ell$ and predict the resistance and support levels. Their existence is proved and nonlinear integral equations characterise them uniquely. Under positive drift the stopping set is a band between two curves, and under negative drift a single boundary.
\end{abstract}

\medskip

\section{Introduction}

This paper extends the optimal prediction approach to resistance and support levels developed by De Angelis and Peskir \cite{optimal_prediction_od_resistence_and_support_level}. In that framework, resistance and support levels are viewed as hidden targets generated by the trading behaviour of market participants. The modelling idea is based on the aspiration level hypothesis \cite{Simon1955_A_Behavioral_Model_of_Rational_Choice}: a trader entering the market has in mind a target price at which he or she is willing to trade in the future. Since such targets are not directly observable and may vary across traders, the aspiration level is modelled as a random variable. The resulting problem is then formulated as a finite-horizon optimal stopping problem.

De Angelis and Peskir studied this problem under the assumption that the asset price follows a geometric Brownian motion. For a suitable class of admissible aspiration-level laws, they proved that the optimal trading boundaries can be interpreted as conditional median curves of the hidden target. The aim of the present paper is to investigate how this theory changes when the geometric Brownian motion model is replaced by a constant elasticity of variance process.

More precisely, we assume that the asset price follows the CEV diffusion
\begin{equation}\label{intro:cev}
    dX_t=\mu X_t\,dt+\sigma X_t^{\beta+1}\,dB_t,
\end{equation}
where $\mu\in\mathbb{R}$, $\sigma>0$, and $\beta>0$. The case $\beta=0$ reduces to the geometric Brownian motion case considered in \cite{optimal_prediction_od_resistence_and_support_level}. The CEV model is a natural extension because it allows volatility to depend on the current price level. It has been widely used in option pricing and in the modelling of level-dependent volatility effects \cite{the_constant_elasticity_of_variance_option_pricing_model}. Further developments and applications of the CEV model can be found in \cite{The_valuation_of_options_for_alternative_stochastic_processes}, \cite{Further_Results_on_the_Constant_Elasticity_of_Variance_Call_Option_Pricing_Model}, and \cite{Cev_chapter}.

The extension from geometric Brownian motion to the CEV process is not a direct substitution. The geometric Brownian motion case has an explicit exponential solution and simple martingale properties, and these facts are useful in the analysis of the stopping set. In the CEV setting, such tools are no longer available in the same form. Although the CEV process has useful connections with Bessel processes \cite{Cev_chapter} and related diffusion transformations \cite{A_jump_to_default_extended_CEV_model}, it does not have the same simple pathwise representation as geometric Brownian motion.

We keep the hidden-target formulation of \cite{optimal_prediction_od_resistence_and_support_level}. The aspiration level $\ell$ is assumed to be non-negative, independent of the price process, and distributed according to a given law. As in the geometric Brownian motion case, the original prediction problem can be reduced to an optimal stopping problem with a convex loss function determined by the distribution function of $\ell$. The median interval of $\ell$ remains important, but the CEV dynamics lead to a different structure of the stopping region.

A key difference from the geometric Brownian motion case is the shape of the stopping region. Under suitable admissibility assumptions on the law of the aspiration level, we show that the positive drift case leads to a stopping band. More precisely, when $\mu>0$, there exist two time-dependent boundaries $a$ and $c$ such that stopping is optimal for prices lying between them, while continuation is optimal both below the lower boundary and above the upper boundary. This contrasts with the one-sided boundary obtained in the geometric Brownian motion setting of \cite{optimal_prediction_od_resistence_and_support_level}. The upper continuation region is a distinctive feature of the positive-elasticity CEV model. In particular, the lower boundary converges to a finite level at maturity, whereas the upper boundary tends to infinity. When $\mu<0$, the stopping region is one-sided and is described by a single increasing boundary.

The optimal boundaries are characterised by nonlinear integral equations. In the positive drift case, the two boundaries solve a coupled system of integral equations; in the negative drift case, the single boundary solves one integral equation. We prove the existence, continuity and uniqueness of the corresponding boundaries. The proof uses methods of optimal stopping and free-boundary theory, together with smooth fit in \cite{optimal_stopping_and_free_boundary_problem} and the local time-space formula from \cite{A_Change-of-Variable_Formula_with_Local_Time_on_Curves}.

We include numerical examples for admissible aspiration-level laws. These examples illustrate the shape of the optimal boundaries and show how the nonlinear integral equations can be solved numerically. The numerical section is used to support the theoretical results rather than to propose a separate numerical method.

The comparative statics further show that the elasticity coefficient has a significant effect on the location of the optimal boundaries. In the positive drift case, decreasing $\beta$ pushes the upper boundary $c(t)$ upward and enlarges the stopping band; in particular, the numerical results provide evidence that $c(t)\to\infty$ for each fixed $t$ as $\beta\to 0$, which is consistent with the GBM case. In the negative drift case, decreasing $\beta$ raises the support boundary $b(t)$ and therefore enlarges the one-sided stopping region. These observations highlight the role of the level-dependent volatility in the CEV model in shaping the optimal prediction of resistance and support levels. Fuller details will be given in Section \ref{sec:comparative} below.

The paper is organized as follows. Section \ref{sec:cev} recalls the CEV process and its relevant properties. Section \ref{sec:formulation} formulates the optimal prediction problem and reduces it to an optimal stopping problem. Section \ref{sec:solution} solves the free-boundary problem and derives the integral equations for the optimal boundaries. Section \ref{sec:examples} presents numerical examples and Section \ref{sec:comparative} gives the comparative statics analysis.

\section{The constant elasticity of variance process}\label{sec:cev}

In this section we introduce the \textit{Constant Elasticity of Variance (CEV) process} and its properties that are useful to solve the problem.

1. \emph{Definition and origin}. The CEV process was first introduced by Cox \cite{the_constant_elasticity_of_variance_option_pricing_model} in work dating back to 1975 in order to capture the leverage effect, i.e., the empirically observed inverse relation between the stock price level and its volatility. In Cox's original parameterisation $dS_t=\mu S_t\,dt+\sigma S_t^{\theta/2}\,dB_t$ the exponent was restricted to $\theta\in[0,2)$, which corresponds to $\beta=(\theta-2)/2\in[-1,0)$ in the convention of the defining equation below; the model was later extended by Emanuel and MacBeth \cite{Further_Results_on_the_Constant_Elasticity_of_Variance_Call_Option_Pricing_Model} to the case $\theta>2$, i.e., $\beta>0$. Finally, Yuen et al. \cite{ESTIMATION_IN_THE_CONSTANT_ELASTICITY_OF_VARIANCE_MODEL} were the first to remove the restriction on the elasticity parameter. We give a definition of a CEV process in the following.

\begin{definition}[\textbf{Constant Elasticity of Variance process}] \label{def 2.1}
Let $(\Omega,\mathcal F,(\mathcal F_t)_{t\ge0},\mathbb P)$ be a filtered probability space satisfying the usual conditions, and let $B=(B_t)_{t\ge0}$ be a standard $(\mathcal F_t)$-Brownian motion. Fix parameters $\mu\in\mathbb R$, $\sigma>0$, $\beta\in\mathbb R$, and an initial value $x>0$.

A continuous, nonnegative, $(\mathcal F_t)$-adapted process $X=(X_t)_{t\ge0}$ is called a \emph{Constant Elasticity of Variance (CEV) process} with parameters $(\mu,\sigma,\beta,x)$ if $X_0=x$ and, for every $t<\tau_0$, where $\tau_0:=\inf\{s\ge0: X_s=0\}$ (with $\tau_0=\infty$ when $0$ is never attained),
\begin{equation} \label{2.1}
    dX_t=\mu X_t\,dt+\sigma X_t^{\beta+1}\,dB_t .
\end{equation}
If the boundary point $0$ is attainable, we additionally impose that it is absorbing. Namely, on $\{\tau_0<\infty\}$ we set $X_t=0$ for all $t\ge\tau_0$.
\end{definition}

2. \emph{Bessel modification}. According to Linetsky and Mendoza's result \cite{Cev_chapter}, the CEV process with parameters $(\mu,\sigma,\beta,x)$ defined as in Definition \ref{def 2.1} with $\beta>0$ has a Bessel modification whose expression is given by
\begin{equation}\label{2.2}
    X_t=e^{\mu t}\left(\sigma |\beta| R_{T_t}^{(\nu)}\right)^{-\frac{1}{\beta}},
\end{equation}
where $R_t^{(\nu)}$ is a Bessel process with index $\nu=1/(2\beta)$, $R_0^{(\nu)}:=r_0=x^{-\beta}/(\sigma |\beta|)$, and $T_t:=(e^{2\mu \beta t}-1)/(2 \mu \beta)$ is a deterministic time and when \(\mu=0\), \(T_t\) is understood as its limiting value \(t\). By the above Bessel modification, one can also derive the CEV density by using the well-known Bessel density:
\begin{equation}\label{2.3}
    p^{(\mu)}(x,z;t)=e^{-\mu t}\,p^{(0)}\!\left(x,e^{-\mu t}z;T_t\right),
\end{equation}
where
\begin{equation}\label{2.4}
\begin{aligned}
    p^{(0)}(x,z;t)
    &=
    \frac{z^{-2\beta-\frac{3}{2}}\,x^{\frac{1}{2}}}{\sigma^2|\beta|\,t}
    I_{|\nu|}\!\left(
    \frac{x^{-\beta}z^{-\beta}}{\sigma^2\beta^2 t}
    \right)
    \exp\!\left(
    -\frac{x^{-2\beta}+z^{-2\beta}}{2\sigma^2\beta^2 t}
    \right).
\end{aligned}
\end{equation}
Here $I_{|\nu|}$ is the modified Bessel function of the first kind of order $\nu$. Moreover, the expectation of a CEV process with positive elasticity coefficient has an explicit expression:
\begin{equation}\label{2.5}
\mathbb{E}_x[X_t]
= x e^{\mu t}
\left(
1 - Q\!\left(
\nu,
\frac{\mu x^{-2\beta}}
{\sigma^2 \beta \left(e^{2\mu\beta t}-1\right)}
\right)
\right),
\end{equation}
where
\begin{equation}\label{2.6}
Q(\nu,x)
=
\frac{1}{\Gamma(\nu)}
\int_x^\infty u^{\nu-1}e^{-u}\,du .
\end{equation}
Based on the Bessel modification above, we can conclude that the CEV process is a Feller process, hence has strong Markov property.

3. \emph{Boundary behaviours}. Based on theorems from Feller's work \cite{Diffusion_processes_in_one_dimension} and methods in Borodin and Salminen's book \cite{handbook_of_brownian_motion}, the nature of the boundary at the origin depends on the parameter $\beta$. It is classified as a regular boundary for $\beta < -1/2$, an exit boundary for $-1/2 \le \beta < 0$, and a natural boundary for all other cases, namely $\beta \ge 0$. On the other hand, the boundary at infinity is an entrance boundary for $\beta>0$ and a natural boundary for $\beta \le 0$. Since the boundary behaviour becomes more complicated when $\beta<0$, the present paper focuses on the case $\beta \ge 0$ and leaves negative elasticity constants for further investigation.

4. \emph{Martingale properties of CEV process}. It can be proved, by Kotani's theorem (see e.g. \cite{Kotani's_theorem}), that the discounted process
\(e^{-\mu t}X_t\) is a non-negative strict local martingale when \(\beta>0\). This result is
also supported by several papers such as \cite{Simulation_of_the_CEV_process_and_the_local_martingale_property}. Moreover, a CEV process with positive elasticity parameter is a supermartingale when $\mu<0$, but in the positive drift case one can only use local submartingale type estimates. The strict local martingale property is the main difference between the present paper and \cite{optimal_prediction_od_resistence_and_support_level}. We will mainly discuss the case $\mu>0$, $\beta>0$ in the present paper. The case $\mu<0$, $\beta>0$ can be treated similarly, and is easier since the CEV process is a supermartingale under these coefficients.

5. \emph{Solutions of CEV SDE}. The solution of the CEV stochastic differential equation depends on the value of $\beta$. For $\beta=0$, the model becomes geometric Brownian motion and has a pathwise unique strong solution. For $\beta>0$, the coefficients are locally Lipschitz on $(0,\infty)$, so the equation has a pathwise unique strong solution on the positive state space. The case $\beta=-1$ also gives a unique strong solution, since the diffusion coefficient is constant. For $\beta\in[-1/2,0)$, weak existence follows from the Engelbert--Schmidt criterion \cite{On_solutions_of_one-dimensional_stochastic_differential_equations_without_drift}, and pathwise uniqueness follows from the Yamada--Watanabe criterion \cite{On_the_uniqueness_of_solutions_of_stochastic_differential_equations}; therefore a pathwise unique strong solution exists in this range. For $\beta<-1/2$ with $\beta\neq-1$, weak solutions exist, and the coefficients are again locally Lipschitz on $(0,\infty)$, so that pathwise uniqueness holds up to the hitting time of $0$; since $0$ is a regular boundary in this range, the process is determined only once a boundary behaviour is prescribed --- under the absorbing convention of Definition \ref{def 2.1} the solution is unique, whereas without a prescribed boundary behaviour uniqueness genuinely fails for $\beta\in(-1,-1/2)$. In particular, for every $\beta\ge0$ the equation \eqref{2.1} has a pathwise unique strong solution, so that the prediction problem studied below is well posed.

\section{Formulation of the problem}\label{sec:formulation}

In financial mathematics, resistance and support levels are the price levels at which most traders prefer to sell and buy, respectively. In practice, however, these levels are difficult to identify since they emerge from the conflicting actions of buyers and sellers and are not directly observable. Instead, they act as hidden targets, which is the terminology introduced by \cite{optimal_detection_of_a_hidden_target}. To formalize this, we employ the aspiration level hypothesis (see \cite{Simon1955_A_Behavioral_Model_of_Rational_Choice}), which states that traders enter the market with a target price in mind at which they are willing to trade. Furthermore, we introduce the concept of the \textit{representative trader}, since it is unrealistic to assume identical aspiration levels across all traders. Assume the aspiration level of this representative trader to be a random variable $\ell \geq 0$, independent of $X$, with continuous distribution function $F$ satisfying $F(0)=0$ and finite expectation, i.e., $\E \ell<\infty$. This leads naturally to the following optimal prediction problem
\begin{equation}\label{3.1}
    V_*(x)=\inf_{0 \leq \tau \leq T} \E_{x}|X_{\tau}-\ell|.
\end{equation}
Here $X$ is a CEV process defined as in Definition \ref{def 2.1}, $\tau$ is any stopping time of $X$ bounded by a given and fixed maturity time $T>0$, and $\E_{x}$ represents the expectation taken under the measure $\P_{x}$ under which $\P_x(X_0=x)=1$. Intuitively, we aim to find the time when the asset price is closest to the representative trader's aspiration level. The nearer a trade occurs to this level, the closer the trade is to the representative target. A stopping time, denoted by $\tau_*$, is considered to be optimal if it achieves the infimum in \eqref{3.1}. Note that the aspiration level $\ell$ \textbf{may} correspond to a support or resistance level, but it is \textbf{not necessarily} one of them. Hence, we treat $\ell$ as an important reference point to predict where the levels will be rather than a real resistance or support level. In practice, the parameters of CEV processes can be inferred from past price and volume data or other available information, such as insider knowledge, by statistical or numerical methods.

We can reduce our problem by the following lemma.

\begin{lemma}\label{lem:loss}
The following identity holds
\begin{equation}\label{3.2}
    \E|x-\ell|=2G(x) + \E \ell
\end{equation}
for all $x>0$, where
\begin{equation}\label{3.3}
    G(x)=\int_0^x \left(F(y)-\frac{1}{2}\right)dy
\end{equation}
is defined as the loss function. Moreover, $G(x)$ can be written in the following expression:
\begin{equation}\label{3.4}
    G(x)=\frac{x}{2}-\E \ell + R(x),
\end{equation}
where $R(x)=\int_x^{\infty} (1-F(y))dy$.
\end{lemma}

\begin{proof}
The identity \eqref{3.2} is a result in \cite{optimal_prediction_od_resistence_and_support_level}, and \eqref{3.4} follows from \eqref{3.3}:
\begin{align}
    G(x)&=\int_0^x \left(F(y)-\frac{1}{2}\right)dy
    =\int_0^x\left(1-(1-F(y))-\frac{1}{2}\right)dy \notag \\
    &=\int_0^x \left(\frac{1}{2}-(1-F(y))\right)dy
    =\frac{x}{2}-\int_0^x(1-F(y))dy \notag \\
    &=\frac{x}{2}-\int_0^{\infty}(1-F(y))dy+
    \int_x^{\infty}(1-F(y))dy \notag \\
    &=\frac{x}{2}-\E \ell +R(x). \notag
\end{align}
This proves the lemma.
\end{proof}

Putting \eqref{3.1}, \eqref{3.2} and \eqref{3.3} together, and using the independence of $\ell$ and $X$, we now have $V_*(x)=2 V(x)+\E \ell$, where
\begin{equation}\label{3.5}
    V(x)=\inf_{0 \leq \tau \leq T} \E_{x} G(X_{\tau}).
\end{equation}
Since $\E \ell$ is a constant, if we find the stopping time $\tau_*$ such that the infimum of \eqref{3.5} is attained, then the infimum of \eqref{3.1} is also attained. In other words, solving the optimal prediction problem \eqref{3.1} is equivalent to solving the optimal stopping problem \eqref{3.5}, which means we have reduced our problem.

By \eqref{3.3}, we see that the derivative of $G$, $G'(s)=F(s)-1/2$, is increasing on $(0,\infty)$, implying that $G(s)$ is convex. The solution is naturally related to the concept of the medians of $\ell$. All medians form a closed interval $[m,M]$ where $m$ and $M$ denote the lowest median and the highest median respectively, with uniqueness occurring if and only if $m=M$. It is clear that $G'=F-1/2<0$ on $(0,m)$, $G'=F-1/2>0$ on $(M,\infty)$ and $G'=0$ on $[m,M]$. This implies that $G$ is strictly decreasing on $(0,m)$, constant on $[m,M]$ and strictly increasing on $(M,\infty)$. To allow the price process to start from any state at any time in $[0,T]\times(0,\infty)$, let us extend \eqref{3.5} to the following form:
\begin{equation} \label{3.6}
    V(t,x)=\inf_{0 \leq \tau \leq T-t} \E_{t,x}G(X_{t+\tau}).
\end{equation}
Here $\E_{t,x}$ is the expectation taken under the measure $\P_{t,x}$ under which $\P_{t,x}(X_t=x)=1$ for any $(t,x)\in[0,T]\times(0,\infty)$.

By general optimal stopping theory (see e.g. \cite{optimal_stopping_and_free_boundary_problem}), define our continuation set $C$ and stopping set $D$ respectively by
\begin{equation}\label{3.7}
    C=\left\{(t,x)\in [0,T)\times(0,\infty)\mid V(t,x)<G(x) \right\}
\end{equation}
and
\begin{equation}\label{3.8}
    D=\left\{(t,x)\in [0,T]\times(0,\infty)\mid V(t,x)=G(x) \right\}.
\end{equation}
The first entry time of $X$ into the stopping set $D$, defined by
\begin{equation}\label{3.9}
    \tau_D=\inf\left\{u \in [0, T-t]\mid (t+u,X_{t+u})\in D \right\},
\end{equation}
is optimal in problem \eqref{3.6}, hence also optimal in problem \eqref{3.1}. Indeed, to see this, note that minimising $\E_{t,x}G(X_{t+\tau})$ is equivalent to minimising the non-negative loss $\E_{t,x}[G(X_{t+\tau})+\E\ell]$ (recall $G\ge-\E\ell$ by \eqref{3.4}), and the value process of this problem is sandwiched between $0$ and the uniformly integrable martingale $\E_{t,x}[G(X_T)+\E\ell\mid\mathcal F_{t+u}]$, obtained by stopping at the terminal time; this sandwich supplies the uniform integrability required in the standard proof of the optimality of the first entry time (see \cite{optimal_stopping_and_free_boundary_problem}, Chapter I, Section 2). We remark that the two-sided condition $\E_{t,x}[\sup_{0\le s\le T-t}|G(X_{t+s})|]<\infty$, under which such results are usually stated, fails here for $t<T$: since $G(x)\ge x/2-\E\ell$, it would force $\E_{t,x}[\sup_{0\le s\le T-t}X_{t+s}]<\infty$, and then dominated convergence would make the discounted price a true martingale, contradicting its strict local martingale property.

We will now introduce the \textit{admissible aspiration level laws} and auxiliary functions, which are useful for the later analysis.

\begin{definition}[\textbf{Admissible Aspiration Level Laws}] \label{adm law}
For $\mu>0$, $\sigma>0$ and $\beta>0$, we let $\mathcal A_{\mathrm{CEV}}(\mu,\sigma,\beta)$ denote the family of continuous, piecewise $C^1$ (i.e., $C^1$ off a locally finite set) probability distribution functions $F$ on $\mathbb R$, with locally bounded density $F'$ and $F(0)=0$, for which there exists $\alpha\in(0,m)$ such that
\begin{equation}\label{3.10}
    x^{2\beta+1}F'(x)
    <
    \frac{\mu}{\sigma^2/2}
    \left(
    \frac{1}{2}-F(x)
    \right)
    \quad \text{for } x\in(0,\alpha),
\end{equation}
and
\begin{equation}\label{3.11}
    x^{2\beta+1}F'(x)
    >
    \frac{\mu}{\sigma^2/2}
    \left(
    \frac{1}{2}-F(x)
    \right)
    \quad \text{for } x\in(\alpha,m),
\end{equation}
where $m$ is the lowest median of $F$.

For $\mu<0$, $\sigma>0$ and $\beta>0$, we let $\mathcal A_{\mathrm{CEV}}(\mu,\sigma,\beta)$ denote the family of continuous, piecewise $C^1$ (i.e., $C^1$ off a locally finite set) probability distribution functions $F$ on $\mathbb R$, with locally bounded density $F'$ and $F(0)=0$, for which there exists $\gamma\in(M,\infty)$ such that
\begin{equation}\label{3.12}
    x^{2\beta+1}F'(x)
    >
    \frac{\mu}{\sigma^2/2}
    \left(
    \frac{1}{2}-F(x)
    \right)
    \quad \text{for } x\in(M,\gamma),
\end{equation}
and
\begin{equation}\label{3.13}
    x^{2\beta+1}F'(x)
    <
    \frac{\mu}{\sigma^2/2}
    \left(
    \frac{1}{2}-F(x)
    \right)
    \quad \text{for } x\in(\gamma,\infty),
\end{equation}
where $M$ is the highest median of $F$.
\end{definition}

We will make use of the following auxiliary functions:
\begin{align}
J(t,x) &= \E_{t,x}G(X_T)=\mathbb{E}_x\left[G\left(X_{T-t}\right)\right]
       = \int_0^\infty G(z)\,p^{(\mu)}(x,z;T-t)\,dz, \label{3.14}\\
H(x) &= \mu x\left(F(x)-\frac{1}{2}\right)
       + \frac{\sigma^2}{2}x^{2\beta+2}F'(x), \label{3.15}\\
L(s,x,y)&=\mathbb{E}_x\left[H(X_s)\1_{\{X_s<y\}}\right]
=\int_0^y H(u)p^{(\mu)}(x,u;s)\,du, \label{3.16} \\
M(s,x,y,z) &= \mathbb{E}_x\left[H(X_s)\1_{\{y<X_s<z\}}\right]
       = \int_y^z H(u)\,p^{(\mu)}(x,u;s)\,du. \label{3.17}
\end{align}

A natural intuition for the stopping set is that once the process starts in it, it is optimal to stop immediately, i.e., waiting will not give a better result. The continuation set has the opposite meaning. Based on this idea, we compare $\E_{t,x}G(X_{t+s})$ and $\E_{t,x}G(X_t)=G(x)$ for $s\in(0,T-t]$, with equality at $s=0$. We first treat the case $\mu>0$, and then the case $\mu<0$. For the positive drift case we have the following analysis.

\begin{lemma} \label{lem 3.2}
Let $X$ be defined as in Definition \ref{def 2.1} with strictly positive $\mu$ and $\beta$. Then the following identity holds:
\begin{equation}\label{3.18}
    \lim_{x \rightarrow \infty} \E_{t,x} X_{t+s}=\frac{e^{\mu s}z_s^{\nu}}{\Gamma(\nu+1)},
\end{equation}
for any given and fixed $t\in[0,T)$ and $s\in(0,T-t]$, where $z_s=\mu e^{-2\mu \beta s}/(\sigma^2 \beta (1-e^{-2\mu \beta s}))$ and $\nu=1/(2\beta)$.
\end{lemma}

\begin{proof}
Let $z_s=\mu e^{-2\mu \beta s}/(\sigma^2 \beta (1-e^{-2\mu \beta s}))$. Then \eqref{2.5} becomes
\begin{equation}\label{3.19}
\mathbb{E}_{t,x}[X_{t+s}]
= x e^{\mu s}
\left(
1 - Q\!\left(
\nu,
z_sx^{-2\beta}
\right)
\right).
\end{equation}
Moreover, $z_sx^{-2\beta} \rightarrow 0$ as $x \rightarrow \infty$ since we only consider positive elasticity constants here. Hence, using expressions in \cite{NIST_handbook_of_mathematical_functions}, we have
\begin{equation}\label{3.20}
    1-Q(\nu,z_s x^{-2\beta})=P(\nu,z_s x^{-2\beta})=(z_s x^{-2\beta})^{\nu}\gamma^*(\nu,z_s x^{-2\beta}) \sim \frac{(z_s x^{-2\beta})^{\nu}}{\Gamma(\nu+1)}
\end{equation}
as $x \rightarrow \infty$. Here \(P(\nu,x)=1-Q(\nu,x)\) denotes the regularised lower incomplete gamma function, and \(\gamma^*(\nu,x)\) denotes the corresponding normalised lower gamma function. Hence, for any fixed $t\geq 0$ and $s\in(0,T-t]$,
\begin{equation}\label{3.21}
\E_{t,x}X_{t+s}=xe^{\mu s} P(\nu, z_s x^{-2\beta})\sim \frac{xe^{\mu s}(z_s x^{-2\beta})^{\frac{1}{2\beta}}}{\Gamma(\nu+1)}=\frac{e^{\mu s} z_s^{\nu}} {\Gamma(\nu+1)}
\end{equation}
as $x\rightarrow\infty$. This proves the lemma. One can also use L'H\^opital's rule to prove it.
\end{proof}

For any given and fixed $s\in(0,T-t]$, Lemma \ref{lem 3.2} implies that there exists some $K_s<\infty$ such that for all $x>K_s$, we have $\E_{t,x}(X_{t+s})<x-2\E \ell$. We can also see that if the CEV process starts at a large enough state, then the expected loss in the future is less than the beginning loss at some future time $s\in(0,T-t]$. Indeed, by \eqref{3.4} and the fact that $0 \leq R(x) \leq \E \ell$ for all $x>0$, we have
\begin{equation}\label{3.22}
\begin{aligned}
    \E_{t,x}G(X_{t+s})&=\frac{1}{2}\E_{t,x}(X_{t+s})-\E\ell+\E_{t,x}R(X_{t+s})\\
    &\leq \frac{1}{2}\E_{t,x}(X_{t+s}) <\frac{1}{2}x-\E \ell \leq G(x)
\end{aligned}
\end{equation}
for some given and fixed $s\in(0,T-t]$ and $x>K_s$. Here $K_s$ is only a sufficiently large threshold depending on the chosen time length $s$, obtained from the limiting estimate, and not the optimal stopping boundary. Since all deterministic times are stopping times, we now have $V(t,x)\leq \E_{t,x}G(X_{t+s})<G(x)$. This means the action of waiting is better than stopping immediately. This result suggests that the sufficiently high price side should belong to the continuation set $C$ under the case of strictly positive drift and elasticity constant.

We claim that $[0,T]\times[m,M]\subseteq D$ and $[0,T)\times(0,\alpha)\subseteq C$, where $\alpha$ is the point appearing in Definition \ref{adm law}. Indeed, by the shape of $G$, the function $G$ is strictly decreasing on $(0,m)$, constant on $[m,M]$, and strictly increasing on $(M,\infty)$. Hence, for any fixed $x\in[m,M]$, we have $G(y)\geq G(x)$ for all $y>0$. Therefore, for any stopping time $\tau$ of $X$, it follows that $\E_{t,x}G(X_{t+\tau})\geq G(x)$. Consequently, by the definition of $V(t,x)$, we obtain $V(t,x)=\inf_{0 \leq \tau \leq T-t} \E_{t,x}G(X_{t+\tau}) \geq G(x)$. On the other hand, by taking $\tau=0$, we have $V(t,x)\leq G(x)$. Combining the above inequalities, we conclude that $V(t,x)=G(x)$ for any $(t,x)\in[0,T]\times[m,M]$, which means $[0,T]\times[m,M]\subseteq D$.

It remains to prove that $[0,T)\times(0,\alpha)\subseteq C$. Define the stopping time $\sigma_{\alpha}:= \inf \{u\geq 0\mid X_{t+u} \geq \alpha \} \wedge (T-t)$. By the continuity of $u\mapsto X_{t+u}$, we always have $0<X_{t+u}\leq \alpha$ for any $u\in[0,\sigma_{\alpha}]$ and $x\in(0,\alpha]$. Applying It\^o--Tanaka's formula and replacing $u$ by $\sigma_{\alpha}$, we have
\begin{equation}\label{3.23}
    G(X_{t+\sigma_{\alpha}})=G(x)+\int_0^{\sigma_{\alpha}}H(X_{t+u})du+
    \int_0^{\sigma_{\alpha}}\sigma X_{t+u}^{\beta+1}G'(X_{t+u})dB_{t+u}.
\end{equation}
Taking $\E_{t,x}$ on both sides gives
\begin{equation}\label{3.24}
\begin{aligned}
    \E_{t,x}G(X_{t+\sigma_{\alpha}})&=G(x)+\E_{t,x}\int_0^{\sigma_{\alpha}}H(X_{t+u})du\\
    &=G(x)+\E_{t,x}\int_0^{\infty}H(X_{t+u})\1_{\{ u \leq \sigma_{\alpha} \} }du.
\end{aligned}
\end{equation}
Here $M_{\sigma_{\alpha}}:=\int_0^{\sigma_{\alpha}}\sigma X_{t+u}^{\beta+1}G'(X_{t+u})dB_{t+u}$ vanishes after taking expectation because
\begin{equation}\label{3.25}
\E_{t,x}\left\langle M, M\right\rangle_{\sigma_{\alpha}}
=\E_{t,x} \int_0^{\sigma_{\alpha}} \sigma^2 X_{t+u}^{2\beta+2}(G'(X_{t+u}))^2du
\leq \frac{\sigma^2 \alpha^{2\beta+2}} {4}(T-t)<\infty.
\end{equation}
By Definition \ref{adm law} we know that $H(x)<0$ when $x\in(0,\alpha)$. Therefore, using \eqref{3.24}, we can conclude that $V(t,x)\leq \E_{t,x}G(X_{t+\sigma_{\alpha}})<G(x)$. Thus, for all $(t,x)\in[0,T)\times(0,\alpha)$, it is better to wait, which is equivalent to saying that $[0,T)\times(0,\alpha)\subseteq C$.

For the case $\mu<0$, the logic is simpler. The process $X$ satisfies \eqref{2.1}. Therefore, for any bounded stopping time $\tau\leq T-t$, we can write
\begin{equation}\label{3.26}
X_{t+\tau}
=
x+\int_0^{\tau} \mu X_{t+s}\,ds
+
\int_0^{\tau} \sigma X_{t+s}^{\beta+1}\,dB_{t+s} .
\end{equation}
Since $\mu<0$ and $X_t>0$, the finite variation part $\int_0^{\tau} \mu X_{t+s}\,ds$ is decreasing. The stochastic integral part $\int_0^{\tau} \sigma X_{t+s}^{\beta+1}\,dB_{t+s}$ is a local martingale. Hence $X$ is a non-negative local supermartingale, and therefore it is a supermartingale. Thus $\mathbb{E}_{t,x} X_{t+\tau}\leq x$.

We now use the shape of the loss function $G$. Recall that $G$ is decreasing on $(0,m)$, attains its minimum on $[m,M]$, and is convex. Let $x\in(0,m)$ and $t\in[0,T]$. As $G$ is decreasing on $(0,m)$ and using Jensen's inequality, we have
\begin{equation}\label{3.27}
G(x)\leq G\left(\mathbb{E}_{t,x}X_{t+\tau}\right)\leq
\mathbb{E}_{t,x}G(X_{t+\tau}).
\end{equation}
Hence immediate stopping is optimal for $x\in(0,m)$. For $x\in[m,M]$ and $t\in[0,T]$, the function $G$ attains its minimum on $[m,M]$. Hence, for all $y>0$, $G(x)\leq G(y)$. Based on this fact, we have $G(x)\leq \mathbb{E}_{t,x}G(X_{t+\tau})$. Thus immediate stopping is also optimal for $x\in[m,M]$. Therefore, when $\mu<0$, $[0,T]\times(0,M]\subseteq D$.

The preceding analysis completes the formulation of the problem. We have reduced the original optimal prediction problem to the optimal stopping problem \eqref{3.6}, introduced the continuation and stopping sets, and obtained preliminary information on their structure from the median interval of $\ell$ and the sign of $H=\mathbb{L}_XG$. These facts will serve as the starting point for the free-boundary analysis below.

\section{Solution to the problem}\label{sec:solution}

In this section we state the main result for the reduced optimal stopping problem. The theorem gives the characterization of the value function and the optimal stopping boundary for both cases $\mu>0$ and $\mu<0$. The proof, however, will be presented only for the case $\mu>0$ with $\beta>0$. This is the more delicate regime, since the positive-elasticity CEV process exhibits strict local martingale behaviour and its high-state behaviour differs from the geometric Brownian motion case. The case $\mu<0$ with $\beta>0$ will be stated without a separate proof. In that case the CEV process is a non-negative supermartingale, which gives stronger estimates, and the corresponding arguments follow by analogous and simpler modifications.

\begin{theorem}\label{thm:main}
Consider the problems \eqref{3.1} and \eqref{3.6}, where $X$ is the CEV process solving \eqref{2.1} with $\beta>0$, and $\ell\geq 0$ satisfying $\mathbb{E}\ell<\infty$ is independent of $X$. Suppose that the distribution function $F$ of $\ell$ belongs to the admissible class $\mathcal A_{\mathrm{CEV}}(\mu,\sigma,\beta)$ for $\mu\neq 0$.

If $\mu>0$, then the stopping set in \eqref{3.6} is given by $D=\{(t,x)\in[0,T)\times(0,\infty)\mid a(t)\leq x\leq c(t)\}\cup(\{T\}\times(0,\infty))$, where the optimal stopping boundaries $a,c:[0,T)\to(0,\infty)$ can be characterised as the unique continuous solution to the coupled nonlinear integral equations
\begin{equation}
\label{4.1}
J(t,a(t))=G(a(t))+\int_t^T M(s-t,a(t),a(s),c(s))\,ds,
\end{equation}
and
\begin{equation}
\label{4.2}
J(t,c(t))=G(c(t))+\int_t^T M(s-t,c(t),a(s),c(s))\,ds,
\end{equation}
satisfying $\alpha\leq a(t)\leq m\leq M\leq c(t)<\infty$ for $t\in[0,T)$, with $a(T-)=\alpha$ and $c(T-)=\infty$. The stopping time
\begin{equation}
\label{4.3}
\tau_{a,c}=\inf\{t\in[0,T)\mid a(t)\leq X_t\leq c(t)\} \wedge T
\end{equation}
is optimal in \eqref{3.1}. The value function $V$ from \eqref{3.6} admits the representation
\begin{equation}
\label{4.4}
V(t,x)=J(t,x)-\int_t^T M(s-t,x,a(s),c(s))\,ds
\end{equation}
for $(t,x)\in[0,T]\times(0,\infty)$. Consequently, the value $V_*(x)$ from \eqref{3.1} is given by $V_*(x)=2V(0,x)+\mathbb{E}\ell$ for $x>0$.

If $\mu<0$, then the stopping set in \eqref{3.6} is given by $D=\{(t,x)\in[0,T)\times(0,\infty)\mid x\leq b(t)\}\cup(\{T\}\times(0,\infty))$, where the optimal stopping boundary $b:[0,T)\to(0,\infty)$ can be characterised as the unique continuous increasing solution to the nonlinear integral equation
\begin{equation}
\label{4.5}
J(t,b(t))=G(b(t))+\int_t^T L(s-t,b(t),b(s))\,ds,
\end{equation}
satisfying $M\leq b(t)\leq \gamma$ for $t\in[0,T)$, with $b(T-)=\gamma$. The stopping time
\begin{equation}
\label{4.6}
\tau_b=\inf\{t\in[0,T)\mid X_t\leq b(t)\} \wedge T
\end{equation}
is optimal in \eqref{3.1}. The value function $V$ from \eqref{3.6} admits the representation
\begin{equation}
\label{4.7}
V(t,x)=J(t,x)-\int_t^T L(s-t,x,b(s))\,ds
\end{equation}
for $(t,x)\in[0,T]\times(0,\infty)$. Consequently, the value $V_*(x)$ from \eqref{3.1} is given by $V_*(x)=2 V(0,x)+\mathbb{E}\ell$ for $x>0$.
\end{theorem}

\begin{proof}
The proof follows the methods explained in \cite{optimal_stopping_and_free_boundary_problem} and will be presented in several steps. As mentioned above, we will only focus on the case of positive drift and positive elasticity coefficient, and give the proof of this case in full detail. The case $\mu<0$ can be proved by analogous and simpler modifications because of the supermartingale property. Hence we assume throughout that $\mu>0$ is given and fixed.

1. We show that $(t,x)\mapsto V(t,x)$ is continuous on $[0,T]\times(0,\infty)$. For this, we first prove that $x\mapsto V(t,x)$ is continuous uniformly for all $t\in[0,T]$, and then show that $t\mapsto V(t,x)$ is continuous. Choosing $x$ and $y$ arbitrarily, and using $|G(u)-G(v)|\leq |u-v|/2$, we have
\begin{align} \label{4.8}
|V(t,y)-V(t,x)|&=\left|\inf_{0 \leq \tau \leq T-t}\E G(X_{t+\tau}^{t,y})-
\inf_{0 \leq \tau \leq T-t}\E G(X_{t+\tau}^{t,x})\right| \notag \\
&\leq \sup_{0 \leq \tau \leq T-t}\left|\E G(X_{t+\tau}^{t,y})-
\E G(X_{t+\tau}^{t,x})\right| \notag \\
&\leq \sup_{0 \leq \tau \leq T-t}\E \left|G(X_{t+\tau}^{t,y})- G(X_{t+\tau}^{t,x})\right| \notag \\
&\leq \sup_{0 \leq \tau \leq T-t} \frac{1}{2}\E\left|X_{t+\tau}^{t,y}-X_{t+\tau}^{t,x}\right|
\leq \frac{1}{2} e^{\mu T}|y-x| \rightarrow 0
\end{align}
as $|y-x|\rightarrow0$. The last inequality is obtained as follows. Define $D_s=X_{t+s}^{t,y}-X_{t+s}^{t,x}$. Without loss of generality, assume $y\geq x$. By pathwise uniqueness and the comparison property for the CEV SDE, we have $X_{t+s}^{t,y} \geq X_{t+s}^{t,x}$ and therefore $D_s\geq0$. Moreover,
\begin{equation}\label{4.9}
    dD_s=\mu D_s\,ds+\sigma\left[(X_{t+s}^{t,y})^{\beta+1}-(X_{t+s}^{t,x})^{\beta+1}\right]dB_{t+s},
\end{equation}
and $\widetilde{D}_s:=e^{-\mu s}D_s$ is a non-negative continuous local martingale, hence a supermartingale. Finally, the optional sampling theorem gives
\begin{equation}\label{4.10}
    \sup_{0 \leq \tau \leq T-t}\E D_{\tau}
    =\sup_{0 \leq \tau \leq T-t}\E e^{\mu \tau}\widetilde{D}_{\tau}
    \leq e^{\mu T} \E \widetilde{D}_0
    \leq e^{\mu T}|x-y|.
\end{equation}

We then show that $t\mapsto V(t,x)$ is continuous. Note that $V(t,x)$ is increasing in $t$. We claim that, for any $h\in[0,T-t]$, $t\in[0,T)$ and $\delta>0$,
\begin{equation}\label{4.11}
0 \leq V(t+h,x)-V(t,x) \leq \frac{1}{2}e^{\mu T} \delta +(\E\ell+G(x))\P_{t,x} \left(\sup_{0 \leq s \leq h} |X_{t+s}^{t,x}-x|>\delta \right),
\end{equation}
and, for any $h\in[0,t]$, $t\in(0,T]$ and $\delta>0$,
\begin{equation}\label{4.12}
\begin{aligned}
0 \leq V(t,x)-V(t-h,x) \leq{}& \frac{1}{2}e^{\mu T} \delta\\
&+(\E\ell+G(x))\P_{t-h,x} \left(\sup_{0 \leq s \leq h} |X_{t-h+s}^{t-h,x}-x|>\delta\right).
\end{aligned}
\end{equation}
We only give the detailed proof of \eqref{4.11}, since \eqref{4.12} can be derived in the same way. Since $\{\tau\geq h\}\in \mathcal{F}_{t+h}$, and in this case $\tau-h$ is admissible after $t+h$, the strong Markov property together with the definition of $V$ gives
\begin{equation}\label{4.13}
    \E\left[G(X_{t+h+(\tau-h)}^{t,x}) \mid \mathcal{F}_{t+h} \right]
    \geq V\left(t+h,X_{t+h}^{t,x}\right).
\end{equation}
Using \eqref{4.8}, we get
\begin{align} \label{4.14}
\E G(X_{t+\tau}^{t,x})
&=\E G(X_{t+\tau}^{t,x})\1_{\{\tau <h\}}+
\E G(X_{t+\tau}^{t,x})\1_{\{\tau \geq h\}} \notag \\
&= \E G(X_{t+\tau}^{t,x})\1_{\{\tau <h\}} +
\E \left(\1_{\{\tau \geq h\}} \E\left[G(X_{t+h+(\tau-h)}^{t,x}) \mid \mathcal{F}_{t+h} \right]\right) \notag \\
&\geq\E\left(G(X_{t+\tau}^{t,x})\1_{\{\tau <h\}}+V\left(t+h,X_{t+h}^{t,x}\right)\1_{\{\tau \geq h\}} \right) \notag \\
&\geq \E \left(V(t+h,x)-\frac{1}{2}e^{\mu T} \delta -
(\E\ell +G(x))\1_{\{\sup_{0 \leq s \leq h} |X_{t+s}^{t,x}-x|>\delta\}} \right) \notag \\
&\geq V(t+h,x)-\frac{1}{2}e^{\mu T} \delta \notag \\
&\quad-
(\E\ell +G(x))\P_{t,x}\left(\sup_{0 \leq s \leq h} |X_{t+s}^{t,x}-x|>\delta\right).
\end{align}
Taking the infimum over all stopping times $\tau\in[0,T-t]$ on the left-hand side of \eqref{4.14}, we obtain \eqref{4.11}. For each fixed $\delta>0$, by the continuity of the sample paths of $X$, we have
\begin{equation}\label{4.15}
\P_{t,x}\left(\sup_{0 \leq s \leq h}|X_{t+s}^{t,x}-x|>\delta\right)\to0
\end{equation}
as $h\downarrow0$. Hence, by \eqref{4.11},
\begin{equation}\label{4.16}
\limsup_{h\downarrow0}\left[V(t+h,x)-V(t,x)\right]\leq \frac{1}{2}e^{\mu T}\delta .
\end{equation}
Since $\delta>0$ is arbitrary, letting $\delta\downarrow0$ gives
\begin{equation}\label{4.17}
\limsup_{h\downarrow0}\left[V(t+h,x)-V(t,x)\right]\leq0.
\end{equation}
On the other hand, using $V(t+h,x)-V(t,x)\geq0$ together with \eqref{4.17}, we have
\begin{equation}\label{4.18}
\lim_{h\downarrow0}\left[V(t+h,x)-V(t,x)\right]=0,
\end{equation}
which proves the right-continuity of $t\mapsto V(t,x)$. The left-continuity follows from \eqref{4.12} in exactly the same way. Since $t\in[0,T]$ is arbitrary, $t\mapsto V(t,x)$ is continuous. Thus the value function $(t,x)\mapsto V(t,x)$ is continuous on $[0,T]\times(0,\infty)$.

2. We show that the stopping set in \eqref{3.6} is given by
\begin{equation}\label{4.19}
D=\{(t,x)\in[0,T)\times(0,\infty)\mid a(t)\le x\le c(t)\}
  \cup(\{T\}\times(0,\infty)),
\end{equation}
where $a:[0,T)\to(0,\infty)$ is decreasing and $c:[0,T)\to(0,\infty)$ is increasing, satisfying $\alpha\le a(t)\le m\le M\le c(t)<\infty$ for every $t\in[0,T)$.

For any fixed $t\in[0,T)$, define $D_t:=\{x>0\mid (t,x)\in D\}$ and $C_t:=\{x>0\mid (t,x)\in C\}$. From the previous analysis, we have already proved that $[m,M]\subseteq D_t$ for every $t\in[0,T]$, hence $D_t\neq\varnothing$ for every $t\in[0,T)$. Moreover, we have proved that $(0,\alpha)\subseteq C_t$, hence $D_t\cap(0,\alpha)=\varnothing$. Besides, by the high-state continuation result, for any fixed $t\in[0,T)$, there exists some $K_t>M$ such that $[K_t,\infty)\subseteq C_t$. Therefore, by \eqref{3.8},
\begin{equation}\label{4.20}
D_t\subseteq[\alpha,K_t].
\end{equation}

We first prove that $D_t$ is compact. Since $V$ and $G$ are continuous, the function $\Phi(t,x):=V(t,x)-G(x)$ is continuous on $[0,T]\times(0,\infty)$. Therefore,
\begin{equation}\label{4.21}
D=\{(t,x)\in[0,T]\times(0,\infty)\mid \Phi(t,x)=0\}=\Phi^{-1}(\{0\}).
\end{equation}
Hence $D$ is closed in $[0,T]\times(0,\infty)$ since $\{0\}$ is closed in $\mathbb{R}$. For any fixed $t\in[0,T)$, define $\eta_t:(0,\infty)\to[0,T]\times(0,\infty)$ by $\eta_t(x):=(t,x)$. This map is continuous and $D_t=\eta_t^{-1}(D)$. Hence $D_t$ is closed in $(0,\infty)$. By \eqref{4.20}, $D_t\subseteq[\alpha,K_t]$, so $D_t$ is a closed subset of the compact interval $[\alpha,K_t]$. Hence $D_t$ is compact.

We next prove that $D_t$ is increasing in $t$. Let $0\le t<s<T$ and suppose that $x\in D_t$. Then $V(t,x)=G(x)$. Since the remaining time is shorter when the starting time is $s$, the value function is increasing in time. Hence $V(s,x)\ge V(t,x)=G(x)$. On the other hand, instantaneous stopping is an admissible strategy, so $V(s,x)\le G(x)$. Putting the two inequalities together gives $V(s,x)=G(x)$. Thus $x\in D_s$, and consequently
\begin{equation}\label{4.22}
D_t\subseteq D_s,
\qquad 0\le t<s<T.
\end{equation}

We now prove that $D_t$ is connected. Fix $t\in[0,T)$ and choose $x_1,x_2\in D_t$ with $x_1<x_2$. It is sufficient to prove that every point between $x_1$ and $x_2$ also belongs to $D_t$. Suppose this is not the case. Then there exists some $y\in(x_1,x_2)$ such that $y\notin D_t$, which means $(t,y)\in C$. Define
\begin{equation}\label{4.23}
\tau_D:=\inf\{u\in[0,T-t]\mid (t+u,X_{t+u})\in D\}.
\end{equation}
Since $C$ is open and $X$ has continuous sample paths, starting from $(t,y)$ gives $\tau_D>0$. Since $x_1,x_2\in D_t$ and $D_t$ is increasing in $t$, by \eqref{4.22} we have $x_1,x_2\in D_{t+u}$ for every $u\in[0,T-t]$. Therefore, before the time $\tau_D$, the process cannot hit $x_1$ or $x_2$. By the continuity of $X$, we have
\begin{equation}\label{4.24}
x_1<X_{t+u}<x_2,
\qquad 0\le u<\tau_D.
\end{equation}
Since $x_1,x_2\in D_t$ and $D_t\cap(0,\alpha)=\varnothing$, we have $x_1\ge\alpha$, and hence $X_{t+u}>\alpha$ for $0\le u<\tau_D$. Moreover, since $[m,M]\subseteq D_{t+u}$ for every $u\in[0,T-t]$, the process cannot enter $[m,M]$ before $\tau_D$. Hence
\begin{equation}\label{4.25}
X_{t+u}\in(\alpha,m)\cup(M,\infty),
\qquad 0\le u<\tau_D.
\end{equation}
On $(\alpha,m)$, the admissibility condition gives $H(x)>0$. On $(M,\infty)$, we have $F(x)-1/2>0$ and $F'(x)\ge0$, which also gives $H(x)>0$. Therefore, $H(X_{t+u})>0$ for $0\le u<\tau_D$, and since $\tau_D>0$,
\begin{equation}\label{4.26}
\mathbb{E}_{t,y}\int_0^{\tau_D}H(X_{t+u})\,du>0.
\end{equation}

Applying It\^o--Tanaka's formula to $G(X_{t+u})$, and then replacing $u$ by $\tau_D$, we get
\begin{equation}\label{4.27}
G(X_{t+\tau_D})
=
G(y)
+
\int_0^{\tau_D}H(X_{t+r})\,dr
+
\int_0^{\tau_D}\sigma X_{t+r}^{\beta+1}G'(X_{t+r})\,dB_{t+r}.
\end{equation}
We now show that the stochastic integral in \eqref{4.27} vanishes after taking expectation. This is the same argument as the one used above when proving $(0,\alpha)\subseteq C_t$. Indeed, before $\tau_D$, the process is trapped in the compact interval $[x_1,x_2]$. Since $|G'|\le1/2$, we have
\begin{equation}\label{4.28}
\int_0^{\tau_D}\sigma^2X_{t+r}^{2\beta+2}(G'(X_{t+r}))^2\,dr
\le
\frac{\sigma^2}{4}x_2^{2\beta+2}(T-t)
<\infty.
\end{equation}
Thus the stopped stochastic integral $\left(\int_0^{u\wedge\tau_D}\sigma X_{t+r}^{\beta+1}G'(X_{t+r})\,dB_{t+r}\right)_{0\le u\le T-t}$ is a square-integrable martingale. Hence
\begin{equation}\label{4.29}
\mathbb{E}_{t,y}
\left[
\int_0^{\tau_D}\sigma X_{t+r}^{\beta+1}G'(X_{t+r})\,dB_{t+r}
\right]
=0.
\end{equation}
Taking $\mathbb{E}_{t,y}$ on both sides of \eqref{4.27}, and using \eqref{4.26} and \eqref{4.29}, we obtain
\begin{equation}\label{4.30}
\mathbb{E}_{t,y}G(X_{t+\tau_D})
=
G(y)
+
\mathbb{E}_{t,y}\int_0^{\tau_D}H(X_{t+r})\,dr
>
G(y).
\end{equation}
Since $\tau_D$ is optimal for the problem starting from $(t,y)$, we have $V(t,y)=\mathbb{E}_{t,y}G(X_{t+\tau_D})>G(y)$. This contradicts the definition of $V$, which gives $V(t,y)\le G(y)$. Hence no such $y$ exists, and therefore $D_t$ is connected.

Now we have proved that, for every $t\in[0,T)$, $D_t$ is non-empty, compact and connected. Hence $D_t$ must be a compact interval. Define $a(t):=\min D_t$ and $c(t):=\max D_t$. Then
\begin{equation}\label{4.31}
D_t=[a(t),c(t)].
\end{equation}
The inclusions proved above give the locations of the two boundaries. Since $D_t\cap(0,\alpha)=\varnothing$, we have $a(t)\ge\alpha$. Since $[m,M]\subseteq D_t=[a(t),c(t)]$, we have $a(t)\le m\le M\le c(t)$. Since $D_t\subseteq[\alpha,K_t]$, we also have $c(t)<\infty$. Therefore,
\begin{equation}\label{4.32}
\alpha\le a(t)\le m\le M\le c(t)<\infty,
\qquad t\in[0,T).
\end{equation}
Finally, since $D_t$ is increasing in $t$, for $0\le t<s<T$ we have $D_t\subseteq D_s$. Using \eqref{4.31}, we get $[a(t),c(t)]\subseteq[a(s),c(s)]$. Hence $a(s)\le a(t)$ and $c(t)\le c(s)$. This means that $a$ is decreasing and $c$ is increasing on $[0,T)$. Combining these conclusions, we obtain \eqref{4.19} as claimed.

3. We then show that $a(t)$ and $c(t)$ are continuous on $[0,T)$, and that $a(T-)=\alpha$, $c(T-)=\infty$. Since $V$ is continuous, $D$ is closed and $C$ is open. 

First we show that $a$ is right-continuous. Consider any time $t\in[0,T)$, and take a sequence $t_n\downarrow t$. Since $a$ is decreasing, the limit $a(t+):=\lim_{n\to\infty}a(t_n)$ exists and satisfies $a(t+)\le a(t)$. Since $(t_n,a(t_n))\in D$ for every $n$, and $D$ is closed, letting $n\to\infty$ gives $(t,a(t+))\in D$. Hence $a(t+)\in D_t=[a(t),c(t)]$, and so $a(t+)\ge a(t)$. Therefore $a(t+)=a(t)$, and $a$ is right-continuous.

Now we show that $a$ has no jump on $(0,T)$. Assume that $a$ has a jump. Since $a$ is decreasing, it can only jump downwards. If we call the jumping time $t_0$, then $a(t_0-)>a(t_0)$. Set
\begin{equation}\label{4.33}
x_1=a(t_0),\qquad x_2=a(t_0-).
\end{equation}
Then $x_1<x_2$. Take any interval $I=(u,v)$ that is strictly contained in $(x_1,x_2)$, i.e., $x_1<u<v<x_2$. Since $\alpha\le a(t)\le m$, we have $I\subset(\alpha,m)$. Since $a$ is monotone, it has left limits and right limits at every time $t$ \cite{Principles_of_mathematical_analysis}. Hence, for $s<t_0$, $a(s)\downarrow a(t_0-)=x_2$ as $s\uparrow t_0$. Therefore there exists some positive $\delta$ such that, for all $s\in(t_0-\delta,t_0)$, we have $a(s)>x_2>v$. Hence, for every $x\in I$, we have $x<a(s)$. By the definition of the continuation set, this gives that $(s,x)$ lies in the continuation set for $s\in(t_0-\delta,t_0)$ and $x\in I$. Thus, on the rectangle $(t_0-\delta,t_0)\times I$, we can use the same test-function argument as in the GBM case.

Take a non-negative function with compact support in $I$, $\varphi\in C_c^2(I)$, $\varphi\not\equiv0$, such that there exists $\varepsilon>0$ satisfying
\begin{equation}\label{4.34}
\varphi(x)=\varphi'(x)=0
\qquad
\text{for all }x\in(0,u+\varepsilon]\cup[v-\varepsilon,\infty).
\end{equation}
Since $t\mapsto V(t,x)$ is increasing, we have $V_t\ge0$. Therefore, for $s\in(t_0-\delta,t_0)$,
\begin{equation}\label{4.35}
0\le V_t(s,x)=-(\mathbb L_XV)(s,x)
\end{equation}
in the continuation region, and consequently
\begin{equation}\label{4.36}
0\le \int_I V_t(s,x)\varphi(x)\,dx
=
-\int_I(\mathbb L_XV)(s,x)\varphi(x)\,dx.
\end{equation}
Moreover,
\begin{equation}\label{4.37}
\int_I(\mathbb L_XV)(s,x)\varphi(x)\,dx
=
\int_I \mu xV_x(s,x)\varphi(x)\,dx
+
\int_I \frac12\sigma^2x^{2\beta+2}V_{xx}(s,x)\varphi(x)\,dx.
\end{equation}
Using integration by parts, and using the fact that $\varphi(x)=\varphi'(x)=0$ near the boundary of $I$, all boundary terms vanish. Hence
\begin{equation}\label{4.38}
\int_I(\mathbb L_XV)(s,x)\varphi(x)\,dx
=
\int_I V(s,x)
\left[
-(\mu x\varphi(x))'
+
\left(\frac12\sigma^2x^{2\beta+2}\varphi(x)\right)''
\right]dx.
\end{equation}
Define
\begin{equation}\label{4.39}
(\mathbb L_X^*\varphi)(x)
=
-(\mu x\varphi(x))'
+
\left(\frac12\sigma^2x^{2\beta+2}\varphi(x)\right)''.
\end{equation}
Then we have
\begin{equation}\label{4.40}
0\le
-\int_I V(s,x)(\mathbb L_X^*\varphi)(x)\,dx.
\end{equation}
Now let $s\uparrow t_0$. Since $V$ is continuous, $V(s,x)\to V(t_0,x)$ for every $x\in I$. Let $K:=\supp\varphi$. Then $K$ is compact. Since $\varphi\in C_c^2(I)$, the functions $\varphi$, $\varphi'$, and $\varphi''$ are continuous and bounded on $K$. Moreover, $x^{2\beta+2}$, $x^{2\beta+1}$, and $x^{2\beta}$ are also continuous and bounded on $K$. Hence $\mathbb L_X^*\varphi$ is bounded on $K$, and $\mathbb L_X^*\varphi=0$ outside $K$. In the same way, $V(s,x)$ is bounded on $[t_0-\delta,t_0]\times K$. Therefore, by dominated convergence,
\begin{equation}\label{4.41}
0\le
-\int_I V(t_0,x)(\mathbb L_X^*\varphi)(x)\,dx.
\end{equation}
At time $t_0$, we have $D_{t_0}=[a(t_0),c(t_0)]$. Since $a(t_0)=x_1<u<v<x_2=a(t_0-)\le m\le c(t_0)$, i.e., $I\subset D_{t_0}$, we have $V(t_0,x)=G(x)$ for all $x\in I$. This gives
\begin{equation}\label{4.42}
0\le
-\int_I G(x)(\mathbb L_X^*\varphi)(x)\,dx
=
-\int_I(\mathbb L_XG)(x)\varphi(x)\,dx
=
-\int_IH(x)\varphi(x)\,dx.
\end{equation}
According to our admissible level laws, $H(x)>0$ for $x\in(\alpha,m)$. Since $I\subset(\alpha,m)$, and since $\varphi\ge0$, $\varphi\not\equiv0$, we have $\int_IH(x)\varphi(x)\,dx>0$. Thus
\begin{equation}\label{4.43}
-\int_IH(x)\varphi(x)\,dx<0,
\end{equation}
which contradicts \eqref{4.42}. Hence $a$ has no jump on $(0,T)$. Combining this with the right-continuity of $a$, we know that $a$ is continuous on $[0,T)$.

The proof that $c$ is continuous is obtained in the same way. Its right-continuity follows from the closedness of $D$. If $c$ had a jump at some time $t_0\in(0,T)$, then, since $c$ is increasing, it could only jump upwards. Setting $x_1=c(t_0-)$ and $x_2=c(t_0)$ so that $x_1<x_2$, we can proceed in the same way as above to derive a contradiction. The only difference is that the interval chosen between $x_1$ and $x_2$ is now contained in $(M,\infty)$, where $H>0$. This proves that $c$ is continuous on $[0,T)$.

To see that $a(T-)=\alpha$, note first that $a(T-)\ge\alpha$ since $[0,T)\times(0,\alpha)$ is contained in $C$ as established above. Assuming that $a(T-)>\alpha$ and setting $x_1=\alpha$ and $x_2=a(T-)$ so that $x_1<x_2$, we can then proceed in exactly the same way as above to derive a contradiction. This shows that $a(T-)=\alpha$ as claimed.

Finally, since $c$ is increasing and $c(t)\ge M$ for all $t\in[0,T)$, the limit $c(T-)$ exists in $[M,\infty]$. Assuming that $c(T-)=c_*<\infty$ and choosing an interval $(x_1,x_2)$ with $c_*<x_1<x_2$, we can again proceed in the same way as above to derive a contradiction. Indeed, this interval is contained in $(M,\infty)$, where $H>0$. This shows that $c(T-)=\infty$ as claimed.

4. We show that the smooth fit holds at $a$ and $c$, meaning that $x\mapsto V(t,x)$ is differentiable at $a(t)$ and $c(t)$ with $V_x(t,a(t))=G'(a(t))$ and $V_x(t,c(t))=G'(c(t))$ for every $t\in[0,T)$. We argue as in \cite{optimal_prediction_od_resistence_and_support_level}, whose method in turn goes back to \cite{optimal_stopping_and_free_boundary_problem}.

Take $t\in[0,T)$ and set $x=a(t)$, assuming first that $a(t)<c(t)$. Since $V(t,y)=G(y)$ for $y\in[a(t),c(t)]$, the right-hand derivative of $V(t,\cdot)$ at $x$ exists and equals $G'(x)$, so it is enough to consider the left-hand derivative. For $\eps\in(0,x)$ we have $V(t,x-\eps)\le G(x-\eps)$ and $V(t,x)=G(x)$, so that
\begin{equation*}
\frac{V(t,x-\eps)-V(t,x)}{-\eps}\ge\frac{G(x-\eps)-G(x)}{-\eps},
\end{equation*}
and letting $\eps\downarrow0$ we obtain $\liminf_{\eps\downarrow0}\big(V(t,x-\eps)-V(t,x)\big)/(-\eps)\ge G'(x)$.

For the reverse inequality, let $\tau_\eps$ denote the optimal stopping time for $V(t,x-\eps)$, i.e.\ the first entry time of $(t+u,X^{t,x-\eps}_{t+u})$ into $D$, and realise $X^{t,x-\eps}$ and $X^{t,x}$ as strong solutions of \eqref{2.1} driven by the same Brownian motion, as in Step 1. Set $D_u:=X^{t,x}_{t+u}-X^{t,x-\eps}_{t+u}$, so that, by the comparison argument of Step 1, $D_u\ge0$ and $\widetilde D_u:=e^{-\mu u}D_u$ is a non-negative local martingale (hence a supermartingale) with $\widetilde D_0=\eps$. Since $\tau_\eps$ is admissible for $V(t,x)$, we have $V(t,x)\le\E\,G(X^{t,x}_{t+\tau_\eps})$ and $V(t,x-\eps)=\E\,G(X^{t,x-\eps}_{t+\tau_\eps})$, whence, by the mean value theorem,
\begin{equation}\label{4.44}
\frac{V(t,x-\eps)-V(t,x)}{-\eps}
\le
\frac{1}{\eps}\,\E\big[G(X^{t,x}_{t+\tau_\eps})-G(X^{t,x-\eps}_{t+\tau_\eps})\big]
=
\frac{1}{\eps}\,\E\big[G'(\xi_\eps)\,D_{\tau_\eps}\big]
\end{equation}
for some $\xi_\eps\in[X^{t,x-\eps}_{t+\tau_\eps},X^{t,x}_{t+\tau_\eps}]$.

We claim that $\tau_\eps\to0$ in $\P$-probability as $\eps\downarrow0$. Let $u_0\in(0,T-t)$ be given and fixed. Since $a$ is non-increasing and $c$ is non-decreasing, at the first time $u\le u_0$ at which $X^{t,x-\eps}_{t+u}\ge x$ we have, by the continuity of the sample paths, $X^{t,x-\eps}_{t+u}=x\in[a(t+u),c(t+u))$, so that $(t+u,X^{t,x-\eps}_{t+u})\in D$ and $\tau_\eps\le u_0$. Hence $\{\tau_\eps>u_0\}\subseteq\{\sup_{0\le u\le u_0}X^{t,x-\eps}_{t+u}<x\}$. Moreover, $X^{t,x-\eps}\ge X^{t,x}-\sup_{0\le u\le u_0}D_u$, and Doob's maximal inequality applied to the non-negative supermartingale $\widetilde D$ gives $\P(\sup_{0\le u\le u_0}D_u>\delta)\le e^{\mu^+(T-t)}\eps/\delta$ for every $\delta>0$, where $\mu^+:=\max\{\mu,0\}$. Therefore
\begin{equation*}
\P(\tau_\eps>u_0)\le\P\Big(\sup_{0\le u\le u_0}X^{t,x}_{t+u}<x+\delta\Big)+e^{\mu^+(T-t)}\,\eps/\delta .
\end{equation*}
Since the diffusion coefficient of $X$ is strictly positive on $(0,\infty)$, the diffusion is regular and $\sup_{0\le u\le u_0}X^{t,x}_{t+u}>x$ $\P$-a.s.\ (see e.g. \cite{handbook_of_brownian_motion}), so letting first $\eps\downarrow0$ and then $\delta\downarrow0$ shows that $\P(\tau_\eps>u_0)\to0$ as $\eps\downarrow0$. In particular, by the continuity of the sample paths, $X^{t,x}_{t+\tau_\eps}\to x$, $D_{\tau_\eps}\to0$ and $\xi_\eps\to x$ in $\P$-probability as $\eps\downarrow0$.

To pass to the limit in \eqref{4.44}, fix $\eta\in(0,x)$ and $N>c(t)$, and let $\sigma$ denote the first time $u$ at which $X^{t,x-\eps}_{t+u}\le\eta$ or $X^{t,x}_{t+u}\ge N$, capped at $T-t$. On $[0,\sigma]$ both processes take values in $[\eta,N]$, where $y\mapsto y^{\beta+1}$ is Lipschitz; a standard argument based on It\^o's formula and Gr\"onwall's inequality therefore yields a constant $C=C(\eta,N,T-t)$ such that $\E[D^2_{\rho\wedge\sigma}]\le C\eps^2$ for every stopping time $\rho\le T-t$, and the stopped process $\widetilde D_{\cdot\wedge\sigma}$ is a true martingale (its integrand being bounded on $[\eta,N]$), so that $\E[\widetilde D_{\tau_\eps\wedge\sigma}]=\eps$. Writing $A:=\{\tau_\eps\le\sigma\}$ and $p_\eps:=\P(A^c)$, we decompose
\begin{equation*}
\begin{aligned}
\frac{1}{\eps}\,\E\big[G'(\xi_\eps)D_{\tau_\eps}\big]
={}& G'(x)\,\frac{1}{\eps}\,\E\big[D_{\tau_\eps}\1_{A}\big]\\
&+\frac{1}{\eps}\,\E\big[(G'(\xi_\eps)-G'(x))\,D_{\tau_\eps}\1_{A}\big]\\
&+\frac{1}{\eps}\,\E\big[G'(\xi_\eps)\,D_{\tau_\eps}\1_{A^c}\big].
\end{aligned}
\end{equation*}
By the Cauchy--Schwarz inequality and the $L^2$-bound above, the second term is bounded in absolute value by $C^{1/2}\big(\E[(G'(\xi_\eps)-G'(x))^2]\big)^{1/2}$, which tends to $0$ as $\eps\downarrow0$, since $G'=F-\tfrac12$ is continuous and bounded and $\xi_\eps\to x$ in probability. For the third term, $|G'|\le\tfrac12$ and, by optional sampling for the supermartingale $\widetilde D$ at $\sigma\le\tau_\eps$ on $A^c$ together with the Cauchy--Schwarz inequality, $\E[D_{\tau_\eps}\1_{A^c}]\le e^{\mu^+(T-t)}\E[\widetilde D_{\sigma}\1_{A^c}]\le C_1\,\eps\,p_\eps^{1/2}$ with a constant $C_1$ depending only on $\eta$, $N$, $T-t$ and the parameters of \eqref{2.1}. For the first term, writing $D_{\tau_\eps}\1_{A}=e^{\mu(\tau_\eps\wedge\sigma)}\widetilde D_{\tau_\eps\wedge\sigma}\1_{A}$ and using $\E[\widetilde D_{\tau_\eps\wedge\sigma}]=\eps$, the same bounds give
\begin{equation*}
\Big|\frac{1}{\eps}\,\E\big[D_{\tau_\eps}\1_{A}\big]-1\Big|
\le
C^{1/2}\big(\E[(e^{\mu(\tau_\eps\wedge\sigma)}-1)^2]\big)^{1/2}+C_1\,p_\eps^{1/2},
\end{equation*}
where the first summand tends to $0$ since $\tau_\eps\wedge\sigma\to0$ in probability. Finally, $p_\eps\le\P(\sigma\le u_0)+\P(\tau_\eps>u_0)$, and by the path continuity and the maximal inequality above, $\limsup_{\eps\downarrow0}\P(\sigma\le u_0)\le\P(\sup_{0\le u\le u_0}X^{t,x}_{t+u}\ge N)+\P(\inf_{0\le u\le u_0}X^{t,x}_{t+u}\le\eta+\delta)$ for every $\delta\in(0,x-\eta)$, and the right-hand side tends to $0$ as $u_0\downarrow0$; hence $p_\eps\to0$ as $\eps\downarrow0$. Combining these facts and letting $\eps\downarrow0$ in \eqref{4.44}, we obtain
\begin{equation*}
\limsup_{\eps\downarrow0}\frac{V(t,x-\eps)-V(t,x)}{-\eps}\le G'(x),
\end{equation*}
which together with the lower bound above shows that the left-hand derivative of $V(t,\cdot)$ at $x$ exists and equals $G'(x)$. Hence $V_x(t,a(t))=G'(a(t))$.

The argument at the upper boundary is symmetric, and we only indicate the changes. For $x=c(t)$ (still with $a(t)<c(t)$) and $\eps>0$, the inequality $V\le G$ with equality at $x$ gives $\limsup_{\eps\downarrow0}(V(t,x+\eps)-V(t,x))/\eps\le G'(x)$. With $\tau_\eps$ the optimal stopping time for $V(t,x+\eps)$ and $D_u:=X^{t,x+\eps}_{t+u}-X^{t,x}_{t+u}\ge0$, the mean value theorem now gives $(V(t,x+\eps)-V(t,x))/\eps\ge\frac{1}{\eps}\E[G'(\xi_\eps)D_{\tau_\eps}]$ for some $\xi_\eps\in[X^{t,x}_{t+\tau_\eps},X^{t,x+\eps}_{t+\tau_\eps}]$. At the first time $u\le u_0$ at which $X^{t,x+\eps}_{t+u}\le x$ we have $X^{t,x+\eps}_{t+u}=x=c(t)\le c(t+u)$ and $x>a(t)\ge a(t+u)$, so that $\{\tau_\eps>u_0\}\subseteq\{\inf_{0\le u\le u_0}X^{t,x+\eps}_{t+u}>x\}$; since $X^{t,x+\eps}\le X^{t,x}+\sup_{0\le u\le u_0}D_u$ and $\inf_{0\le u\le u_0}X^{t,x}_{t+u}<x$ $\P$-a.s.\ by the regularity of $X$, the same use of Doob's maximal inequality yields $\tau_\eps\to0$ in probability, and the localisation argument above goes through verbatim, giving $\liminf_{\eps\downarrow0}(V(t,x+\eps)-V(t,x))/\eps\ge G'(x)$. Hence $V_x(t,c(t))=G'(c(t))$.

If instead $a(t)=c(t)$, which by the inequalities of Step~2 can only happen with $a(t)=c(t)=m=M$, the smooth fit at this point holds by a direct argument: since $G(m)=\min G$, we have $V(t,y)\ge G(m)$ for all $y>0$, while $V\le G$ and $V(t,m)=G(m)$; hence $0\le V(t,m\pm h)-V(t,m)\le G(m\pm h)-G(m)\le h\sup_{|y-m|\le h}|F(y)-\tfrac12|=o(h)$ as $h\downarrow0$, because $F$ is continuous with $F(m)=\tfrac12$, so that $V_x(t,m)$ exists and equals $0=G'(m)$.

5. Combining the above analysis with the Markovian property of the CEV process, we can conclude that $V$, $a$ and $c$ solve the following free-boundary problem:
\begin{align}
&V_t(t,x)+\mu xV_x(t,x)+\frac{\sigma^2}{2}x^{2\beta+2}V_{xx}(t,x)=0
\label{4.45} \\
&\qquad \text{for } x\in(0,a(t))\cup(c(t),\infty),\ t\in[0,T),
\notag \\
&V(t,x)=G(x)
\quad \text{for } x\in[a(t),c(t)],\ t\in[0,T),
\quad \text{(instantaneous stopping)},
\label{4.46} \\
&V_x(t,x)=G'(x)
\quad \text{for } x=a(t)\text{ or }x=c(t),\ t\in[0,T),
\quad \text{(smooth fit)},
\label{4.47} \\
&V(t,x)<G(x)
\quad \text{for } x\in(0,a(t))\cup(c(t),\infty),\ t\in[0,T),
\label{4.48} \\
&V(T,x)=G(x)
\quad \text{for } x\in(0,\infty).
\label{4.49}
\end{align}
We next derive a representation of the value function $V$ through the boundaries $a$ and $c$. The argument relies on the local time-space calculus framework for free-boundary problems, see for instance \cite{optimal_stopping_and_free_boundary_problem} and the literature cited there. This representation will then be used to obtain two coupled nonlinear integral equations, which will be shown to determine $a$ and $c$ uniquely.

6. We prove the representation formula \eqref{4.4} for $V$ and then derive the coupled nonlinear integral equations \eqref{4.1} and \eqref{4.2}. Fix $(t,x)\in[0,T)\times(0,\infty)$, and let $s\in(0,T-t)$ be given and fixed. Since $a$ and $c$ are continuous and monotone on $[0,T)$, they are of bounded variation on every compact sub-interval of $[0,T)$. Hence, on $[t,t+s]$, the two curves $a$ and $c$ are finite and continuous, and they satisfy the boundary requirements of the local time-space formula derived in \cite{A_Change-of-Variable_Formula_with_Local_Time_on_Curves}. When $m<M$ the two curves are non-intersecting, being separated by the median interval. When $m=M$ they could a priori touch; by their monotonicity and continuity this can only happen with $a=c=m$ along an initial time interval, on which the open band $(a(r),c(r))$ is empty, so that the occupation term below receives no contribution from such an interval, while the corresponding local-time term vanishes by the smooth-fit property established in Step~4 for this degenerate configuration as well. Moreover, $V$ is continuous, $C^{1,2}$ in the continuation region, and equal to $G$ in the stopping region. Since $G$ is $C^1$ and piecewise $C^2$, and by \eqref{4.45} and \eqref{4.46} we know that $V_t+\mathbb L_X V$ is locally bounded. The smooth-fit condition at $a$ and $c$ makes the two local time terms vanish. Hence, applying the local time-space formula to $V(t+u,X_{t+u})$, $0\le u\le s$, gives
\begin{equation}\label{4.50}
\begin{aligned}
V(t+s,X_{t+s})
&=
V(t,x)
+
\int_0^s
H(X_{t+u})
\1_{\{a(t+u)<X_{t+u}<c(t+u)\}}
\,du
+
M_s,
\end{aligned}
\end{equation}
where
\begin{equation}\label{4.51}
M_s
=
\int_0^s
\sigma X_{t+u}^{\beta+1}
V_x(t+u,X_{t+u})\,dB_{t+u}
\end{equation}
is a continuous local martingale. Here the value of $V_x$ on the two curves is irrelevant, since the local time terms have already disappeared by smooth fit.

We now localize $M$. Let $\rho_n = \inf\{u\ge0:\langle M,M\rangle_u\ge n\}\wedge s$. Then $(M_{u\wedge\rho_n})_{u\ge0}$ is a square-integrable martingale and $\rho_n\uparrow s$ as $n\to\infty$. Replacing $s$ by $\rho_n$ in \eqref{4.50}, and taking expectation under $\P_{t,x}$, we obtain
\begin{equation}\label{4.52}
\E_{t,x}V(t+\rho_n,X_{t+\rho_n})
=
V(t,x)
+
\E_{t,x}\int_0^{\rho_n}
H(X_{t+u})
\1_{\{a(t+u)<X_{t+u}<c(t+u)\}}
\,du .
\end{equation}
We next pass $n$ to infinity. First note that the integrand in the right hand side of \eqref{4.52} is non-negative. Indeed, on $(a(r),m)$, admissibility in Definition \ref{adm law} gives $H(x)>0$. On $[m,M]$ we have $F(x)=1/2$, so the drift term of $H$ vanishes and $H(x)=\tfrac{\sigma^2}{2}x^{2\beta+2}F'(x)\ge0$, with $H=0$ on the open interval $(m,M)$, where $F$ is constant and hence $F'=0$. On $(M,c(r))$, we have $F(x)-1/2>0$ and $F'(x)\ge0$, hence $H(x)>0$. Therefore $H(x)\1_{\{a(t)<x<c(t)\}}\ge0$ for all $t<T$. Hence, by the monotone convergence theorem,
\begin{multline}\label{4.53}
\E_{t,x}\int_0^{\rho_n}
H(X_{t+u})\1_{\{a(t+u)<X_{t+u}<c(t+u)\}}\,du\\
\longrightarrow
\E_{t,x}\int_0^s
H(X_{t+u})\1_{\{a(t+u)<X_{t+u}<c(t+u)\}}\,du .
\end{multline}

It remains to justify the convergence on the left-hand side of \eqref{4.52}. By \eqref{3.4} and the fact that $0\le R(y)\le \E\ell$, we have $G(x)\ge -\E\ell$ for all $x>0$. By definition of $V$, this implies $V(t,x)\ge -\E\ell$ for all $(t,x)\in[0,T]\times(0,\infty)$. Denote $c_0:=\E\ell$. Then
\begin{equation}\label{4.54}
V(t+\rho_n,X_{t+\rho_n})+c_0\ge0.
\end{equation}
On the other hand, choosing to stop at the terminal time $T$ and using the strong Markov property gives
\begin{equation}\label{4.55}
V(t+\rho_n,X_{t+\rho_n})
\le
\E_{t+\rho_n,X_{t+\rho_n}}G(X_T)
=\E_{t,x}\left[G(X_T)\mid\mathcal F_{t+\rho_n}\right].
\end{equation}
Define $A_u:=\E_{t,x}\left[G(X_T)\mid\mathcal F_{t+u}\right]$ for $0\le u\le T-t$. We have $G(X_T)\in L^1(\P_{t,x})$. Indeed, \eqref{3.4} gives $|G(x)|\le x/2+\E\ell$, and since the discounted CEV process $e^{-\mu u}X_{t+u}$ is a non-negative local martingale, it is a supermartingale. Hence $\E_{t,x}X_T\le e^{\mu(T-t)}x<\infty$. Therefore, $(A_u)_{u\ge0}$ is a uniformly integrable martingale, and so is $(A_{\rho_n})_{n\ge1}$. Combining \eqref{4.54} and \eqref{4.55}, we have $0\le V(t+\rho_n,X_{t+\rho_n})+c_0\le A_{\rho_n}+c_0$. Thus the sequence $\big(V(t+\rho_n,X_{t+\rho_n})\big)_{n\ge1}$ is uniformly integrable. Since $V$ is continuous and $X$ has continuous sample paths, we also have $V(t+\rho_n,X_{t+\rho_n})\to V(t+s,X_{t+s})$. Therefore,
\begin{equation}\label{4.56}
\E_{t,x}V(t+\rho_n,X_{t+\rho_n})
\longrightarrow
\E_{t,x}V(t+s,X_{t+s}).
\end{equation}
Passing $n$ to infinity in both sides of \eqref{4.52}, and putting \eqref{4.53} and \eqref{4.56} together, gives
\begin{equation}\label{4.57}
\E_{t,x}V(t+s,X_{t+s})
=
V(t,x)
+
\E_{t,x}\int_0^s
H(X_{t+u})
\1_{\{a(t+u)<X_{t+u}<c(t+u)\}}
\,du.
\end{equation}
Equivalently,
\begin{equation}\label{4.58}
V(t,x)
=
\E_{t,x}V(t+s,X_{t+s})
-
\E_{t,x}\int_0^s
H(X_{t+u})
\1_{\{a(t+u)<X_{t+u}<c(t+u)\}}
\,du .
\end{equation}
Now let $s=T-t$. By \eqref{3.14} and \eqref{3.17}, we have
\begin{equation}\label{4.59}
V(t,x)
=
J(t,x)
-
\int_t^T
M(s-t,x,a(s),c(s))\,ds,
\end{equation}
which is exactly \eqref{4.4}. Here we used Fubini's theorem and \eqref{3.17}:
\begin{align}\label{4.60}
& \E_{t,x}\int_0^{T-t}
H(X_{t+u})
\1_{\{a(t+u)<X_{t+u}<c(t+u)\}}
\,du \notag \\
&=\int_0^{T-t}
\E_{t,x}\left[
H(X_{t+u})
\1_{\{a(t+u)<X_{t+u}<c(t+u)\}}
\right]du \notag \\
&=\int_0^{T-t}
M(u,x,a(t+u),c(t+u))\,du
=\int_t^T
M(s-t,x,a(s),c(s))\,ds .
\end{align}
Substituting $x=a(t)$ and $x=c(t)$ into \eqref{4.59} gives, respectively,
\begin{equation}\label{4.61}
J(t,a(t))
=
G(a(t))
+
\int_t^T
M(s-t,a(t),a(s),c(s))\,ds,
\end{equation}
and
\begin{equation}\label{4.62}
J(t,c(t))
=
G(c(t))
+
\int_t^T
M(s-t,c(t),a(s),c(s))\,ds,
\end{equation}
which are \eqref{4.1} and \eqref{4.2}, as claimed.

7. We show that $a$ and $c$ are the unique solution to the coupled nonlinear integral equations \eqref{4.1} and \eqref{4.2} in the class of continuous functions $t\mapsto a(t)$ and $t\mapsto c(t)$ satisfying $\alpha\le a(t)\le m\le M\le c(t)<\infty$ for $0\le t<T$.

Let $(\tilde a,\tilde c)$ be another pair of continuous functions on $[0,T)$ satisfying $\alpha\le \tilde a(t)\le m\le M\le \tilde c(t)<\infty$ for $0\le t<T$, and suppose that $(\tilde a,\tilde c)$ solves \eqref{4.1} and \eqref{4.2} respectively, that is,
\begin{equation}\label{4.63}
J(t,\tilde a(t))
=
G(\tilde a(t))
+
\int_t^T
M(s-t,\tilde a(t),\tilde a(s),\tilde c(s))\,ds,
\end{equation}
and
\begin{equation}\label{4.64}
J(t,\tilde c(t))
=
G(\tilde c(t))
+
\int_t^T
M(s-t,\tilde c(t),\tilde a(s),\tilde c(s))\,ds,
\end{equation}
for every $0\le t<T$. We will now prove that $\tilde a(t)=a(t)$ and $\tilde c(t)=c(t)$ for $0\le t<T$.

\emph{(i).} We first show that, for every $(t,x)\in[0,T)\times(0,\infty)$,
\begin{equation}\label{4.65}
\E_{t,x}\int_0^{T-t}H^+(X_{t+u})\,du<\infty,
\end{equation}
where $H^+(x)=\max\{H(x),0\}$, $H^-(x)=\max\{-H(x),0\}$ and $H=H^+-H^-$. Define the localising sequence by $\rho_n=\inf\{u\ge0:X_{t+u}\notin(1/n,n)\}\wedge(T-t)$. Since the CEV process with $\beta>0$ stays in $(0,\infty)$ over every finite time interval, we have $\rho_n\uparrow T-t$ as $n\uparrow\infty$. Applying It\^o--Tanaka's formula to $G(X_{t+u})$ and substituting $u$ by $\rho_n$ gives
\begin{equation}\label{4.66}
G(X_{t+\rho_n})
=
G(x)
+
\int_0^{\rho_n}H(X_{t+u})\,du
+
\int_0^{\rho_n}
\sigma X_{t+u}^{\beta+1}G'(X_{t+u})\,dB_{t+u}.
\end{equation}
There is no local-time term in \eqref{4.66} since the loss function $G$ is $C^1$ and piecewise $C^2$. Moreover, on $[0,\rho_n]$, the process $X$ is constrained inside $[1/n,n]$, and $|G'(x)|=|F(x)-1/2|\le 1/2$. Therefore, $\E_{t,x}\int_0^{\rho_n}\sigma^2X_{t+u}^{2\beta+2}(G'(X_{t+u}))^2\,du\le \sigma^2n^{2\beta+2}(T-t)/4<\infty$, so the stochastic integral in \eqref{4.66} is a martingale. Taking expectations on both sides of \eqref{4.66} gives
\begin{equation}\label{4.67}
\E_{t,x}G(X_{t+\rho_n})
=
G(x)
+
\E_{t,x}\int_0^{\rho_n}H(X_{t+u})\,du .
\end{equation}
Equivalently,
\begin{equation}\label{4.68}
\E_{t,x}\int_0^{\rho_n}H^+(X_{t+u})\,du
=
\E_{t,x}G(X_{t+\rho_n})
-
G(x)
+
\E_{t,x}\int_0^{\rho_n}H^-(X_{t+u})\,du .
\end{equation}
We now bound the right hand side. From \eqref{3.4} we have $G(y)\le y/2$, and, as noted after \eqref{4.55}, the discounted CEV process $e^{-\mu s}X_{t+s}$ is a non-negative supermartingale, so that the optional sampling theorem gives $\E_{t,x}\left(e^{-\mu\rho_n}X_{t+\rho_n}\right)\le x$, thus $\E_{t,x}X_{t+\rho_n}\le e^{\mu(T-t)}x$. So we can conclude that
\begin{equation}\label{4.69}
\E_{t,x}G(X_{t+\rho_n})
\le
\frac12 e^{\mu(T-t)}x.
\end{equation}
It remains to control $H^-$. By the admissibility condition, on $(0,\alpha)$ we have $H(x)<0$. Also, $H(x)\ge -\mu x(1/2-F(x))\ge -\mu\alpha/2$. On $[\alpha,\infty)$, the admissibility condition gives $H\ge0$. Therefore,
\begin{equation}\label{4.70}
0\le H^-(x)\le \frac{\mu\alpha}{2}.
\end{equation}
Combining \eqref{4.68}, \eqref{4.69} and \eqref{4.70}, we obtain
\begin{equation}\label{4.71}
\E_{t,x}\int_0^{\rho_n}H^+(X_{t+u})\,du
\le
\frac12 e^{\mu(T-t)}x
-
G(x)
+
\frac{\mu\alpha}{2}(T-t).
\end{equation}
The right hand side is finite and independent of $n$. Letting $n\to\infty$ and applying monotone convergence gives \eqref{4.65}.

\emph{(ii).} We now construct a martingale that will be used later. For any candidate solution pair $(\tilde a,\tilde c)$ such that $\alpha\le \tilde a(t)\le m\le M\le \tilde c(t)<\infty$, we have
\begin{equation}\label{4.72}
0\le
H(X_{t+u})
\1_{\{\tilde a(t+u)<X_{t+u}<\tilde c(t+u)\}}
\le
H^+(X_{t+u}).
\end{equation}
Hence, by \eqref{4.65},
\begin{equation}\label{4.73}
\E_{t,x}\int_0^{T-t}
H(X_{t+u})
\1_{\{\tilde a(t+u)<X_{t+u}<\tilde c(t+u)\}}
\,du<\infty.
\end{equation}
Now define
\begin{equation}\label{4.74}
V^{\tilde a,\tilde c}(t,x)
=
\E_{t,x}G(X_T)-\E_{t,x}\int_0^{T-t}H(X_{t+u})\1_{\{\tilde a(t+u)<X_{t+u}<\tilde c(t+u)\}}\,du .
\end{equation}
The right hand side is finite by \eqref{4.73} and the fact that $G(X_T)\in L^1(\P_{t,x})$, which was proved in Step 6. Since $(\tilde a,\tilde c)$ solves \eqref{4.63} and \eqref{4.64}, we have
\begin{equation}\label{4.75}
V^{\tilde a,\tilde c}(t,\tilde a(t))=G(\tilde a(t)),
\qquad
V^{\tilde a,\tilde c}(t,\tilde c(t))=G(\tilde c(t)),
\end{equation}
for all $t\in[0,T)$. Moreover, by definition,
\begin{equation}\label{4.76}
V^{\tilde a,\tilde c}(T,x)=G(x),
\qquad \text{for all } x>0.
\end{equation}
Fix $(t,x)\in[0,T)\times(0,\infty)$, and define
\begin{equation}\label{4.77}
Y^{\tilde a,\tilde c}
=
G(X_T)
-
\int_0^{T-t}
H(X_{t+u})\1_{\{\tilde a(t+u)<X_{t+u}<\tilde c(t+u)\}}\,du .
\end{equation}
Then $Y^{\tilde a,\tilde c}\in L^1(\P_{t,x})$. We claim that the process $\mathcal M^{\tilde a,\tilde c}$, defined by
\begin{equation}\label{4.78}
\begin{aligned}
\mathcal M^{\tilde a,\tilde c}_u
:={}&
V^{\tilde a,\tilde c}(t+u,X_{t+u})
-
\int_0^u
H(X_{t+r})\1_{\{\tilde a(t+r)<X_{t+r}<\tilde c(t+r)\}}\,dr,\\
& 0\le u\le T-t,
\end{aligned}
\end{equation}
is a uniformly integrable martingale under $\P_{t,x}$. By \eqref{4.74} and the Markov property we get
\begin{equation}\label{4.79}
\begin{aligned}
&V^{\tilde a,\tilde c}(t+u,X_{t+u})\\
&\quad=
\E_{t,x}\left[
G(X_T)
-
\int_u^{T-t}
H(X_{t+r})\1_{\{\tilde a(t+r)<X_{t+r}<\tilde c(t+r)\}}\,dr
\Bigm|\mathcal F_{t+u}
\right].
\end{aligned}
\end{equation}
Since $\int_0^uH(X_{t+r})\1_{\{\tilde a(t+r)<X_{t+r}<\tilde c(t+r)\}}\,dr$ is $\mathcal F_{t+u}$-measurable, \eqref{4.78} and \eqref{4.79} give
\begin{align}\label{4.80}
\mathcal M^{\tilde a,\tilde c}_u
&=
\E_{t,x}\left[
G(X_T)
-
\int_u^{T-t}
H(X_{t+r})\1_{\{\tilde a(t+r)<X_{t+r}<\tilde c(t+r)\}}\,dr
\Bigm|\mathcal F_{t+u}
\right] \notag \\
&\quad -
\int_0^uH(X_{t+r})\1_{\{\tilde a(t+r)<X_{t+r}<\tilde c(t+r)\}}\,dr \notag \\
&=
\E_{t,x}\left[
G(X_T)
-
\int_0^{T-t}
H(X_{t+r})\1_{\{\tilde a(t+r)<X_{t+r}<\tilde c(t+r)\}}\,dr
\Bigm|\mathcal F_{t+u}
\right] \notag \\
&=
\E_{t,x}\left[
Y^{\tilde a,\tilde c}
\mid\mathcal F_{t+u}
\right].
\end{align}
Since $Y^{\tilde a,\tilde c}\in L^1(\P_{t,x})$, the process $(\mathcal M^{\tilde a,\tilde c}_u)_{0\le u\le T-t}$ is a uniformly integrable martingale. Therefore, for every stopping time $\theta\le T-t$,
\begin{equation}\label{4.81}
\E_{t,x}V^{\tilde a,\tilde c}(t+\theta,X_{t+\theta})
=
V^{\tilde a,\tilde c}(t,x)
+
\E_{t,x}\int_0^\theta
H(X_{t+u})\1_{\{\tilde a(t+u)<X_{t+u}<\tilde c(t+u)\}}\,du .
\end{equation}
By Step 6, the true value function has the representation \eqref{4.4}. Repeating the same conditional expectation argument for the true solution pair $(a,c)$, we have that, for every stopping time $\theta\le T-t$,
\begin{equation}\label{4.82}
\E_{t,x}V(t+\theta,X_{t+\theta})
=
V(t,x)
+
\E_{t,x}\int_0^\theta
H(X_{t+u})
\1_{\{a(t+u)<X_{t+u}<c(t+u)\}}
\,du .
\end{equation}

\emph{(iii).} We now prove that
\begin{equation}\label{4.83}
V^{\tilde a,\tilde c}(t,x)=G(x)
\quad\text{whenever}\quad
\tilde a(t)\le x\le\tilde c(t).
\end{equation}
If $x=\tilde a(t)$ or $x=\tilde c(t)$, this follows from \eqref{4.75}. Hence suppose that $\tilde a(t)<x<\tilde c(t)$, and define $\sigma_{\tilde a,\tilde c}=\inf\{u\in[0,T-t]:X_{t+u}\notin(\tilde a(t+u),\tilde c(t+u))\}$. For $\epsilon\in(0,T-t)$, set $S_\epsilon:=T-\epsilon$ and $\sigma_\epsilon:=\sigma_{\tilde a,\tilde c}\wedge(S_\epsilon-t)$. Then, on the compact interval $[t,S_\epsilon]$, the function $\tilde c$ is continuous and finite. Hence $C_\epsilon:=\sup_{r\in[t,S_\epsilon]}\tilde c(r)<\infty$. Before $\sigma_\epsilon$, the process stays inside the candidate stopping interval, and hence $\1_{\{\tilde a(t+u)<X_{t+u}<\tilde c(t+u)\}}=1$ for all $u\in[0,\sigma_\epsilon)$. Applying \eqref{4.81} with $\theta=\sigma_\epsilon$, we get
\begin{equation}\label{4.84}
V^{\tilde a,\tilde c}(t,x)
=
\E_{t,x}V^{\tilde a,\tilde c}(t+\sigma_\epsilon,X_{t+\sigma_\epsilon})
-
\E_{t,x}\int_0^{\sigma_\epsilon}H(X_{t+u})\,du .
\end{equation}
We also apply It\^o--Tanaka's formula to $G(X_{t+u})$, as in \eqref{4.66}, and then replace $u$ by $\sigma_\epsilon$. On $[0,\sigma_\epsilon]$, we have $\alpha\le X_{t+u}\le C_\epsilon$. Since $|G'|\le1/2$, we have
\begin{equation}\label{4.85}
\E_{t,x}\int_0^{\sigma_\epsilon}
\sigma^2X_{t+u}^{2\beta+2}(G'(X_{t+u}))^2\,du
\le
\frac{\sigma^2}{4}C_\epsilon^{2\beta+2}(T-t)<\infty.
\end{equation}
Therefore the stochastic integral in It\^o--Tanaka's formula is a square-integrable martingale, and taking expectations gives
\begin{equation}\label{4.86}
\E_{t,x}G(X_{t+\sigma_\epsilon})
=
G(x)
+
\E_{t,x}\int_0^{\sigma_\epsilon}H(X_{t+u})\,du .
\end{equation}
Combining \eqref{4.84} and \eqref{4.86}, we get
\begin{equation}\label{4.87}
V^{\tilde a,\tilde c}(t,x)-G(x)
=
\E_{t,x}\left[
V^{\tilde a,\tilde c}(t+\sigma_\epsilon,X_{t+\sigma_\epsilon})
-
G(X_{t+\sigma_\epsilon})
\right].
\end{equation}
On the event $\{\sigma_{\tilde a,\tilde c}\le S_\epsilon-t\}$, the process exits the candidate interval before or at $S_\epsilon$. By continuity of $t\mapsto X_t$, $\tilde a$ and $\tilde c$, it hits one of the candidate curves, and hence \eqref{4.75} gives $V^{\tilde a,\tilde c}(t+\sigma_\epsilon,X_{t+\sigma_\epsilon})=G(X_{t+\sigma_\epsilon})$. Thus \eqref{4.87} becomes
\begin{equation}\label{4.88}
V^{\tilde a,\tilde c}(t,x)-G(x)
=
\E_{t,x}\left[
\left(
V^{\tilde a,\tilde c}(S_\epsilon,X_{S_\epsilon})
-
G(X_{S_\epsilon})
\right)
\1_{\{\sigma_{\tilde a,\tilde c}>S_\epsilon-t\}}
\right].
\end{equation}
We now let $\epsilon\downarrow0$. We claim that
\begin{equation}\label{4.89}
\E_{t,x}\left|
V^{\tilde a,\tilde c}(S_\epsilon,X_{S_\epsilon})
-
G(X_{S_\epsilon})
\right|
\longrightarrow0.
\end{equation}
Indeed, by \eqref{4.74} applied conditionally at time $S_\epsilon$,
\begin{equation}\label{4.90}
V^{\tilde a,\tilde c}(S_\epsilon,X_{S_\epsilon})
=
\E_{t,x}\left[
G(X_T)
-
\int_{S_\epsilon}^{T}
H(X_u)
\1_{\{\tilde a(u)<X_u<\tilde c(u)\}}
\,du
\Bigm|\mathcal F_{S_\epsilon}
\right].
\end{equation}
Since $G(X_{S_\epsilon})$ is $\mathcal F_{S_\epsilon}$-measurable, Jensen's inequality for conditional expectations gives
\begin{align}\label{4.91}
&\E_{t,x}\left|
V^{\tilde a,\tilde c}(S_\epsilon,X_{S_\epsilon})
-
G(X_{S_\epsilon})
\right| \notag \\
&\le \E_{t,x}
\left|G(X_T)-G(X_{S_\epsilon})-
\int_{S_\epsilon}^{T}H(X_u)\1_{\{\tilde a(u)<X_u<\tilde c(u)\}}\,du\right| \notag \\
&\le
\E_{t,x}|G(X_T)-G(X_{S_\epsilon})|+
\E_{t,x}\int_{S_\epsilon}^{T}
H(X_r)
\1_{\{\tilde a(r)<X_r<\tilde c(r)\}}
\,dr .
\end{align}
For the first term, since $|G'|\le1/2$, we have $\E_{t,x}|G(X_T)-G(X_{S_\epsilon})|\le \E_{t,x}|X_T-X_{S_\epsilon}|/2$. The path continuity of $X$ implies that $\E_{t,x}|X_T-X_{S_\epsilon}|\to0$ as $\epsilon\downarrow0$; the proof is analogous to the procedure in Step 1. Therefore, the first term in \eqref{4.91} tends to zero. For the second term, we proved that $\E_{t,x}\int_t^T H^+(X_r)\,dr<\infty$ and $0\le H(X_r)\1_{\{\tilde a(r)<X_r<\tilde c(r)\}}\le H^+(X_r)$. Thus
\begin{equation}\label{4.92}
\E_{t,x}\int_{S_\epsilon}^{T}
H(X_u)
\1_{\{\tilde a(u)<X_u<\tilde c(u)\}}
\,du
\longrightarrow0
\end{equation}
by monotone convergence applied to the decreasing terminal intervals. Combining \eqref{4.91} and \eqref{4.92}, we obtain \eqref{4.89}. Letting $\epsilon\downarrow0$ in \eqref{4.88}, we get \eqref{4.83}.

\emph{(iv).} We next show that
\begin{equation}\label{4.93}
V^{\tilde a,\tilde c}(t,x)\ge V(t,x),
\quad  \text{for all } (t,x)\in[0,T]\times(0,\infty).
\end{equation}
For $x\in[\tilde a(t),\tilde c(t)]$, by \eqref{4.83} we have $V^{\tilde a,\tilde c}(t,x)=G(x)\ge V(t,x)$. If $x\notin[\tilde a(t),\tilde c(t)]$, define
\begin{equation}\label{4.94}
\tau_{\tilde a,\tilde c}
=
\inf\{u\in[0,T-t):
\tilde a(t+u)\le X_{t+u}\le\tilde c(t+u)\}.
\end{equation}
Before $\tau_{\tilde a,\tilde c}$, we have $\1_{\{\tilde a(t+u)<X_{t+u}<\tilde c(t+u)\}}=0$. Hence \eqref{4.81} gives
\begin{equation}\label{4.95}
V^{\tilde a,\tilde c}(t,x)
=
\E_{t,x}
V^{\tilde a,\tilde c}
(t+\tau_{\tilde a,\tilde c},X_{t+\tau_{\tilde a,\tilde c}}).
\end{equation}
At $\tau_{\tilde a,\tilde c}$, the process either enters the candidate stopping interval or reaches $T$. In both cases, by \eqref{4.83} and \eqref{4.76}, respectively, we have
\begin{equation}\label{4.96}
V^{\tilde a,\tilde c}
(t+\tau_{\tilde a,\tilde c},X_{t+\tau_{\tilde a,\tilde c}})
=
G(X_{t+\tau_{\tilde a,\tilde c}}).
\end{equation}
Therefore,
\begin{equation}\label{4.97}
V^{\tilde a,\tilde c}(t,x)
=
\E_{t,x}G(X_{t+\tau_{\tilde a,\tilde c}})
\ge
V(t,x),
\end{equation}
because $\tau_{\tilde a,\tilde c}$ is an admissible stopping time for the problem defining $V(t,x)$.

\emph{(v).} We now compare the candidate boundaries with the true boundaries. Recall that $H(x)>0$ on $(\alpha,m)\cup(M,\infty)$. In the present step, we prove that $\tilde a(t)\le a(t)$ for all $t\in[0,T)$. Suppose this is not true. Then there exists $t<T$ such that $\tilde a(t)>a(t)$.

By continuity of $a$ and $\tilde a$, we may choose a small enough $\epsilon>0$ such that $a(t+u)<\tilde a(t+u)$ for $u\in[0,\epsilon]$, where $\epsilon<T-t$. Set $x=\tilde a(t)$ and define
\begin{equation}\label{4.98}
\theta
=
\inf\{u\in[0,\epsilon]:
X_{t+u}\le a(t+u)
\ \text{or}\
X_{t+u}>c(t+u)\}.
\end{equation}
At the initial point $x=\tilde a(t)$, we have $x\in[a(t),c(t)]\cap[\tilde a(t),\tilde c(t)]$. Therefore, by \eqref{4.83} and by the definition of the true stopping set,
\begin{equation}\label{4.99}
V^{\tilde a,\tilde c}(t,x)=G(x)=V(t,x).
\end{equation}
At time $t+\theta$, we have $V^{\tilde a,\tilde c}\ge V$ by \eqref{4.93}. Hence
\begin{equation}\label{4.100}
0
\le
\E_{t,x}\left[
V^{\tilde a,\tilde c}(t+\theta,X_{t+\theta})
-
V(t+\theta,X_{t+\theta})
\right].
\end{equation}
Using \eqref{4.81}, \eqref{4.82}, \eqref{4.99} and \eqref{4.100}, we obtain
\begin{align}\label{4.101}
0
&\le
\E_{t,x}\int_0^\theta
H(X_{t+u})
\left(
\1_{\{\tilde a(t+u)<X_{t+u}<\tilde c(t+u)\}}
-
\1_{\{a(t+u)<X_{t+u}<c(t+u)\}}
\right)du .
\end{align}
For $0\le u<\theta$, we have $a(t+u)<X_{t+u}\le c(t+u)$, which gives $\1_{\{a(t+u)<X_{t+u}<c(t+u)\}}=1$. Moreover, $\1_{\{\tilde a(t+u)<X_{t+u}<\tilde c(t+u)\}}\le1$. Hence
\begin{equation}\label{4.102}
\1_{\{\tilde a(t+u)<X_{t+u}<\tilde c(t+u)\}}-
\1_{\{a(t+u)<X_{t+u}<c(t+u)\}}
\le
-
\1_{\{a(t+u)<X_{t+u}\le\tilde a(t+u)\}}.
\end{equation}
Thus \eqref{4.101} implies
\begin{equation}\label{4.103}
0
\le
-
\E_{t,x}\int_0^\theta
H(X_{t+u})
\1_{\{a(t+u)<X_{t+u}\le\tilde a(t+u)\}}
\,du .
\end{equation}
The expectation in \eqref{4.103}, without the minus sign, is strictly positive. Indeed, since $a(t)<\tilde a(t)$ and both $a,\tilde a$ are continuous, the strip $\{(u,y):a(t+u)<y\le \tilde a(t+u)\}$ contains, on a small deterministic time interval, a non-empty rectangle contained in $(\alpha,m)$. On this rectangle $H>0$. Moreover, the CEV diffusion is non-degenerate on compact subintervals of $(0,\infty)$, and hence it has positive probability to enter this rectangle and spend a strictly positive amount of Lebesgue time there before $\theta$. Therefore
\begin{equation}\label{4.104}
\E_{t,x}\int_0^\theta
H(X_{t+u})\1_{\{a(t+u)<X_{t+u}\le\tilde a(t+u)\}}\,du>0,
\end{equation}
which contradicts \eqref{4.103}. Hence $\tilde a(t)\le a(t)$ for all $0\le t<T$.

\emph{(vi).} We now prove the reverse inequality for $a(t)$. Suppose that there exists some $t<T$ such that $\tilde a(t)<a(t)$. Choose $x\in(\tilde a(t),a(t))$. Then $x\in[\tilde a(t),\tilde c(t)]$, and hence by \eqref{4.83} we have $V^{\tilde a,\tilde c}(t,x)=G(x)$. On the other hand, since $x<a(t)$, the point $(t,x)$ lies in the true continuation region. Therefore $V(t,x)<G(x)=V^{\tilde a,\tilde c}(t,x)$.

Define
\begin{equation}\label{4.105}
\tau_a=\inf\{u\in[0,T-t]:X_{t+u}\ge a(t+u)\}.
\end{equation}
Before $\tau_a$, the process stays below the true lower boundary, and hence $\1_{\{a(t+u)<X_{t+u}<c(t+u)\}}=0$ for $0\le u<\tau_a$. At time $t+\tau_a$, either the terminal time is reached, or the process hits the true lower boundary. In both cases,
\begin{equation}\label{4.106}
V^{\tilde a,\tilde c}(t+\tau_a,X_{t+\tau_a})
=
V(t+\tau_a,X_{t+\tau_a})
=
G(X_{t+\tau_a}).
\end{equation}
Indeed, if $t+\tau_a=T$, this follows from the terminal condition. If $t+\tau_a<T$, then $X_{t+\tau_a}=a(t+\tau_a)$. By the inequality already proved in \emph{(v)}, $\tilde a(t+\tau_a)\le a(t+\tau_a)$, and since $a(t+\tau_a)\le m\le M\le\tilde c(t+\tau_a)$, the point $a(t+\tau_a)$ belongs to the candidate stopping interval as well as to the true stopping interval. Applying \eqref{4.81} and \eqref{4.82} with $\theta=\tau_a$, and using \eqref{4.106}, we obtain
\begin{equation}\label{4.107}
0
=
V^{\tilde a,\tilde c}(t,x)-V(t,x)
+\E_{t,x}\int_0^{\tau_a}
H(X_{t+u})
\1_{\{\tilde a(t+u)<X_{t+u}<\tilde c(t+u)\}}\,du .
\end{equation}
The first term in \eqref{4.107} is strictly positive by the analysis in the first paragraph. Moreover, the integral term is non-negative. Indeed, on the set where the indicator is equal to one and $u<\tau_a$, we have $\tilde a(t+u)<X_{t+u}<a(t+u)\le m$, while $\tilde a(t+u)\ge\alpha$. Hence $X_{t+u}\in(\alpha,m)$, where $H>0$. Therefore the right hand side of \eqref{4.107} is strictly positive, which is impossible. Hence $\tilde a(t)\ge a(t)$ for $0\le t<T$. Combining this with $\tilde a(t)\le a(t)$, which was proved in \emph{(v)}, gives
\begin{equation}\label{4.108}
\tilde a(t)=a(t), \qquad 0\le t<T.
\end{equation}

The uniqueness of the upper boundary $c(t)$ is obtained in the same way, and we only point out the necessary changes. First, replacing the strip between $a$ and $\tilde a$ in \emph{(v)} by the strip between $\tilde c$ and $c$, and using that $H>0$ on $(M,\infty)$, one obtains $\tilde c(t)\ge c(t)$ for all $0\le t<T$. Conversely, if $\tilde c(t)>c(t)$ for some $t<T$, choosing $x\in(c(t),\tilde c(t))$ and stopping at the first hitting time of the true upper boundary leads, by the same comparison argument based on \eqref{4.81} and \eqref{4.82}, to the corresponding contradiction. Hence $\tilde c(t)=c(t)$. Together with \eqref{4.108}, this proves that every continuous pair $(\tilde a,\tilde c)$ satisfying \eqref{4.1} and \eqref{4.2} must coincide with $(a,c)$ and this completes the proof.
\end{proof}
\section{Examples}\label{sec:examples}

In this section we illustrate the results derived in the previous section through two specific examples. We assume throughout that the asset price is given by the CEV process \eqref{2.1} with $\beta>0$, and that the loss function $G$ is given by \eqref{3.3}, where $F$ denotes the distribution function of the aspiration level $\ell$. The boundary curves displayed below are obtained by solving the nonlinear integral equations in Theorem \ref{thm:main} by a backward-induction discretization. The numerical calculations are used only to plot the examples; any simulated paths shown in the figures are included only to illustrate the first-entry form of the optimal stopping time.

\begin{example}[\textbf{Resistance levels}] \label{resistance}
Consider the case when $\mu>0$ and the aspiration level $\ell$ is exponentially distributed with parameter $\lambda>0$. In the CEV model we take $\mu=0.15$, $\sigma=0.30$, $\beta=0.70$, $\lambda=2$, and $T=1$. Thus $F(x)=1-e^{-2x}$ and $m=M=(\log 2)/2\simeq 0.34657$. A direct numerical calculation gives $\alpha\simeq 0.32581$. By the result above, the optimal stopping time is given by $\tau^*=\inf\{t\in[0,T): a(t)\leq X_t\leq c(t)\}\wedge T$ (as in \eqref{4.3}), where $a$ and $c$ are the lower and upper optimal boundaries. For the parameters above, we have $a(T-)=\alpha\simeq 0.32581$, $c(T-)=+\infty$, $a(0)\simeq 0.34234$, and $c(0)\simeq 2.67449$. The optimal selling action triggered at $\tau^*$ creates a resistance level and pushes the price down at $a(t)$. However, the upper boundary $c(t)$ requires a different interpretation. It should not be viewed as a standard selling level for constructing trading strategies. Instead, if the price process starts above $c(t)$, the model suggests that the current price is already outside the region in which the present prediction framework gives a meaningful resistance-type trading signal. In this sense, $c(t)$ can be interpreted as a model-invalidation boundary from the perspective of trading strategies. Indeed, if $c(t)$ were treated as an optimal selling level and the price started above it, then the trader would have to wait for the price to decrease before selling the asset, which is not consistent with the usual interpretation of a resistance level. The reason is that our problem is not designed to find a maximum-profit trading strategy. Rather, it aims to find an optimal stopping time at which the asset price is as close as possible to the aspiration level $\ell$. Therefore, $c(t)$ is a boundary associated with closeness to the aspiration level, but not necessarily a boundary that maximises trading profit. This point will be explained further in the next section. See Figure 1.

\end{example}

\begin{figure}[htbp]
    \centering
    \begin{tikzpicture}[
        srcbox/.style={draw=black, line width=0.5pt},
        zoombox/.style={draw=black, line width=0.5pt, inner sep=0pt},
        conn/.style={draw=black!55, line width=0.4pt},
    ]
        \node[anchor=south west, inner sep=0] (main) at (0,0)
            {\includegraphics[width=0.57\textwidth]{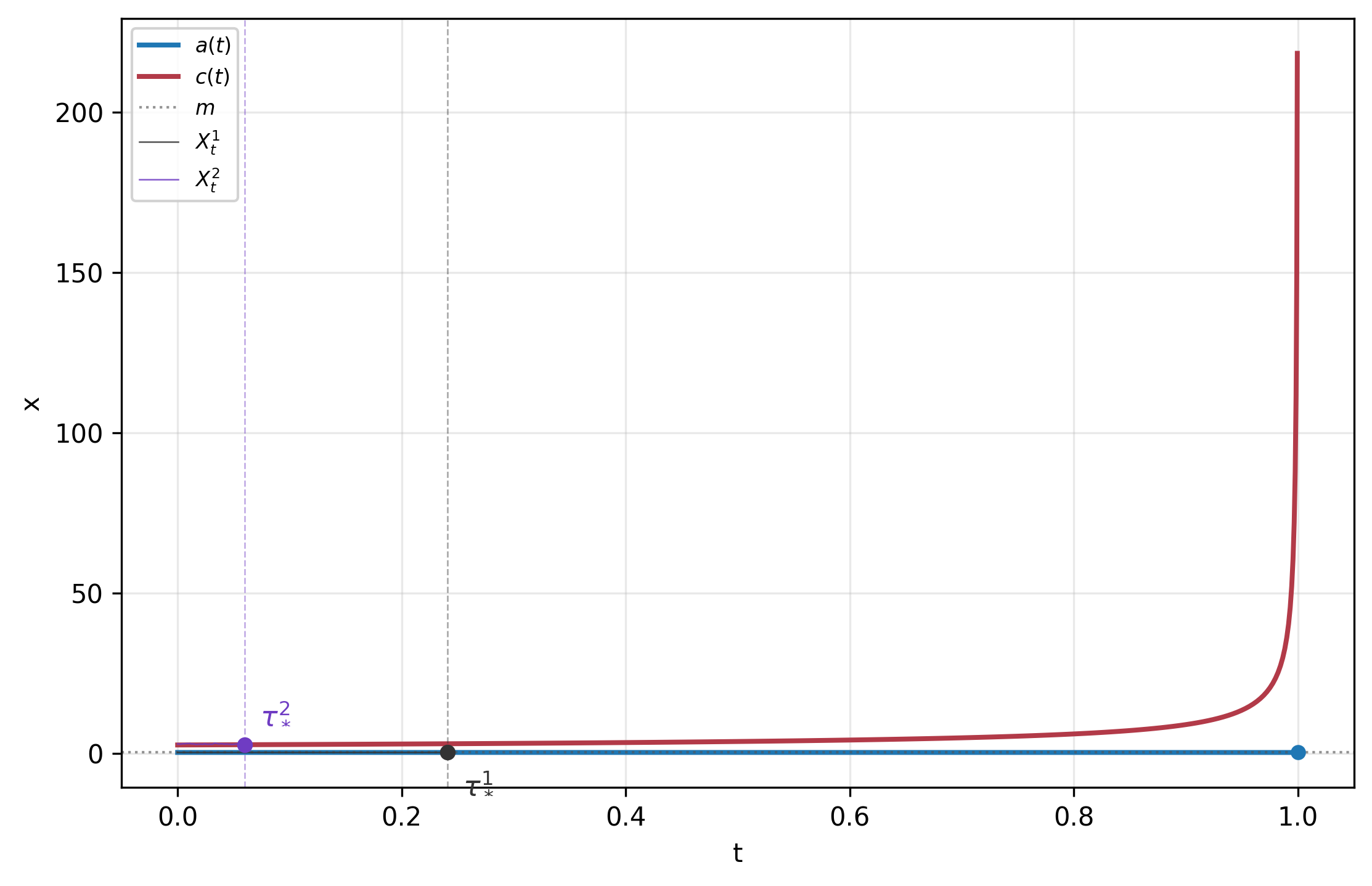}};
        \begin{scope}[x={(main.south east)}, y={(main.north west)}]
            \coordinate (srcSW) at (0.125,0.070);
            \coordinate (srcNE) at (0.340,0.200);
            \node[srcbox, fit=(srcSW)(srcNE), inner sep=0pt] (src) {};
            \node[zoombox, anchor=east] (zoom) at (-0.05,0.46)
                {\includegraphics[width=0.31\textwidth]{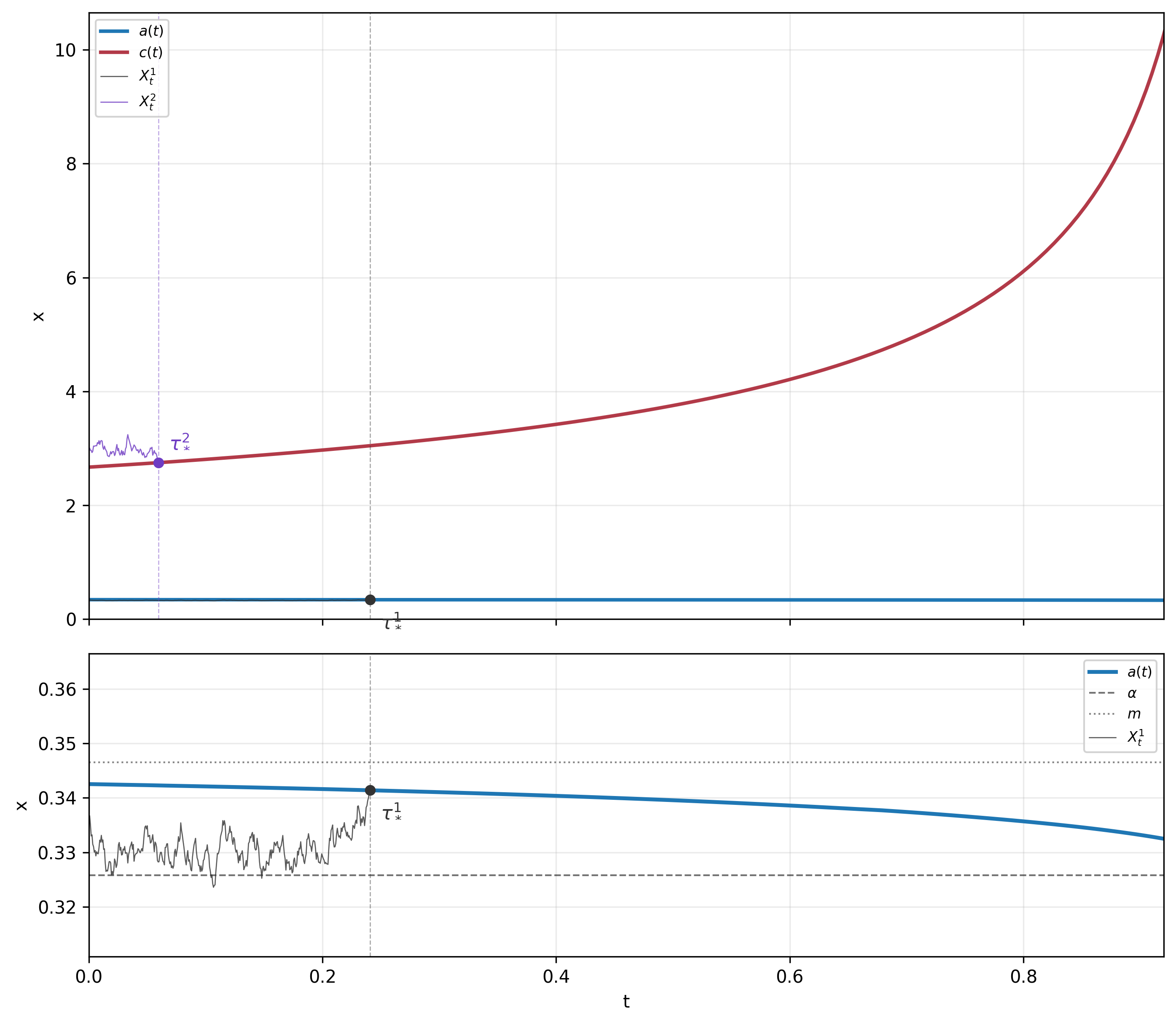}};
            \draw[conn] (src.north west) -- (zoom.north east);
            \draw[conn] (src.south west) -- (zoom.south east);
        \end{scope}
    \end{tikzpicture}
    \caption{The optimal resistance stopping band from Example~\ref{resistance} when the aspiration level $\ell$ is exponentially distributed with parameter $2$. The framed panel enlarges the boxed region: the upper part shows the band between $a(t)$ and $c(t)$ together with two sample paths, and the lower part magnifies the lower boundary $a(t)$ between $\alpha$ and $m$.}
    \label{fig:positive-boundaries}
\end{figure}
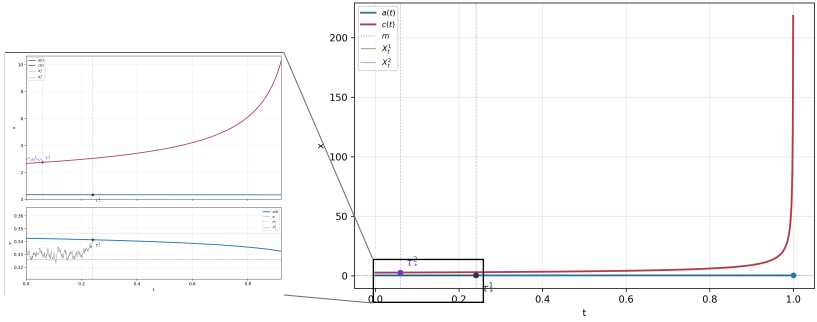

\begin{example}[\textbf{Support levels}] \label{support}
Consider the case when $\mu<0$ and the aspiration level $\ell$ is exponentially distributed with parameter $\lambda>0$. In the CEV model we take $\mu=-0.15$, $\sigma=0.30$, $\beta=0.70$, $\lambda=2$, and $T=1$, so that $F$ and the median $m=M$ are as in Example~\ref{resistance}. A direct numerical calculation gives $\gamma\simeq 0.37414$. By the result above, the optimal stopping time is given by $\tau^*=\inf\{t\in[0,T): X_t\leq b(t)\}\wedge T$ (as in \eqref{4.6}), where $b$ is the optimal support boundary. For the parameters above, we have $b(T-)=\gamma\simeq 0.37414$ and $b(0)\simeq 0.35320$. The optimal buying action triggered at $\tau^*$ creates a support level and pushes the price up. See Figure 2.
\end{example}

\begin{figure}[htbp]
    \centering
    \IfFileExists{mu_negative_boundary.png}{%
    \includegraphics[
        width=0.65\textwidth,
        height=0.25\textheight,
        keepaspectratio
    ]{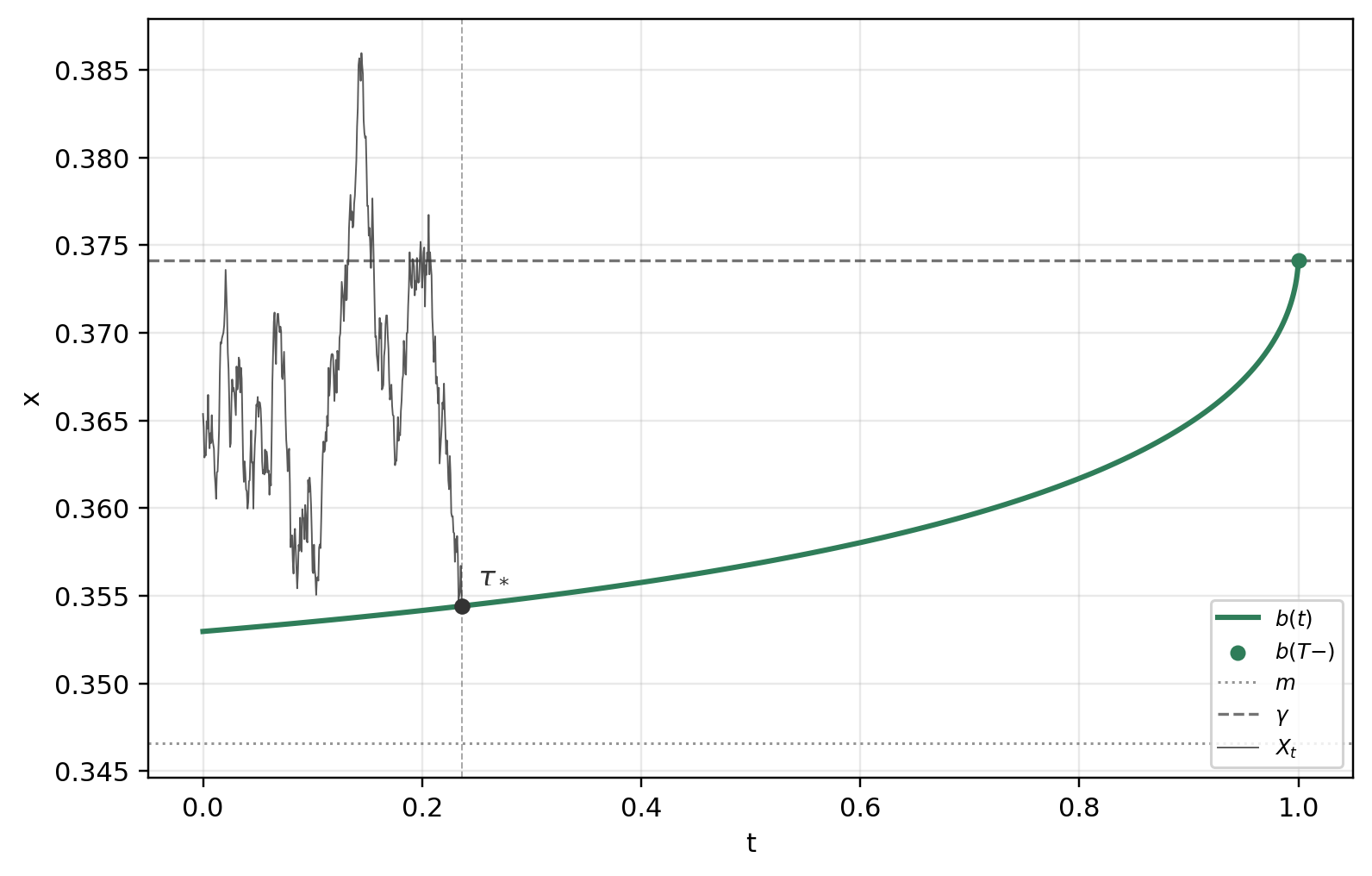}}{%
    \fbox{\parbox[c][0.25\textheight][c]{0.65\textwidth}{\centering \texttt{mu\_negative\_boundary.png}}}}
    \caption{The support level from Example~\ref{support} when the aspiration level $\ell$ is exponentially distributed with parameter $2$.}
    \label{fig:negative-boundary}
\end{figure}

In the next section, we will present some comparative statics.

\section{Comparative Statics} \label{sec:comparative}

For the comparative statics of $\mu>0$ and $\mu<0$, we adopt the parameter settings of example \ref{resistance} and example \ref{support}, respectively, except for the elasticity coefficients. In the positive drift case, we compute 8 pairs of optimal boundaries for $\beta \in [0.01,0.3]$ (figure 3), while in the negative drift case, we compute 10 boundary curves for $\beta \in [0.01,2]$ (figure 4). These numerical results illustrate how the optimal boundaries are affected by $\beta$.

In Figure 3, for every given and fixed $t$, it is clear that $c(t)$ increases dramatically as $\beta$ decreases. Moreover, when $\beta=0.01$, the value of $c(t)$ is too large to be shown in the figure; hence, there are only 7 curves of $c(t)$. This result provides evidence that $c(t) \to \infty$ for every given and fixed $t$ as $\beta \to 0$, which is consistent with the optimal boundaries in the GBM case (\cite{optimal_prediction_od_resistence_and_support_level}). In the enlarged part of Figure 3, for every given and fixed $t$, we see that the value of $a(t)$ decreases as $\beta$ decreases. Combining the above two results, we can conclude that the stopping area increases as $\beta$ decreases, which corresponds to the fact that when $\beta \to 0$, the CEV process will be reduced to a GBM. Financially interpreting this, when the price process starts under the lower resistance level $a(t)$, the process is less than 1, leading to the result that a higher elasticity coefficient leads to lower volatility. This means that, at any given and fixed time, the trader will have more confidence to wait, since the risk is relatively smaller, and that is why we see the behaviours of the lower boundaries.
 For the upper boundary \(c(t)\), it should not be interpreted as a
standard selling level. Instead, it is more appropriate to regard it as a model-invalidation boundary. Recall that, the CEV SDE gives the instantaneous relative volatility $\frac{\sigma X_t^{\beta+1}}{X_t}=\sigma X_t^\beta$.
Hence, in the high-price region \(x>1\), a larger value of \(\beta\) implies a higher volatility as the price level increases. In other words, when \(\beta\) is large, a high price represents a more unstable market state, so that the price has a larger chance to move towards a region closer to the representative aspiration level (recall again that our original optimal prediction problem is to find a time at which the price is closest to the aspiration level $\ell$ rather than finding the profit-maximising trading strategy) .
Therefore, the traders are more willing to wait for a price that is closer to the aspiration level. This leads to a lower upper boundary
\(c(t)\).

On the contrary, when \(\beta\) decreases, the volatility is smaller. High prices are then relatively more persistent under the positive drift \(\mu>0\), while the probability of moving back towards a region closer to the aspiration level is reduced. As a result, waiting becomes less valuable for prices that are high but not extremely
high. Immediate stopping remains optimal over a wider range of high prices, and the upper
edge of the stopping region is pushed upward. This explains why, for each fixed \(t\), the
boundary \(c(t)\) increases as \(\beta\) decreases.

Thus, the observed increase of \(c(t)\) as \(\beta\) decreases reflects the fact that the model
requires a more extreme price level before the upper continuation region becomes optimal. We emphasize again that $c(t)$ is a level such that the price gets close enough to the aspiration level rather than a level that makes the maximum profit.

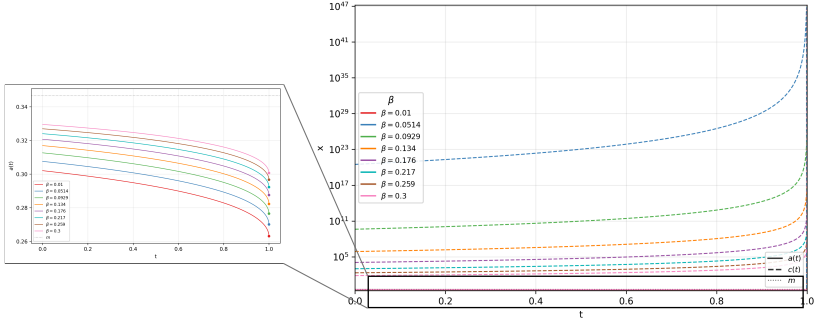
\begin{figure}[htbp]
    \centering
    \begin{tikzpicture}[
        srcbox/.style={draw=black, line width=0.5pt},
        zoombox/.style={draw=black, line width=0.5pt, inner sep=0pt},
        conn/.style={draw=black!55, line width=0.4pt},
    ]
        \node[anchor=south west, inner sep=0] (main) at (0,0)
            {\includegraphics[width=0.57\textwidth]{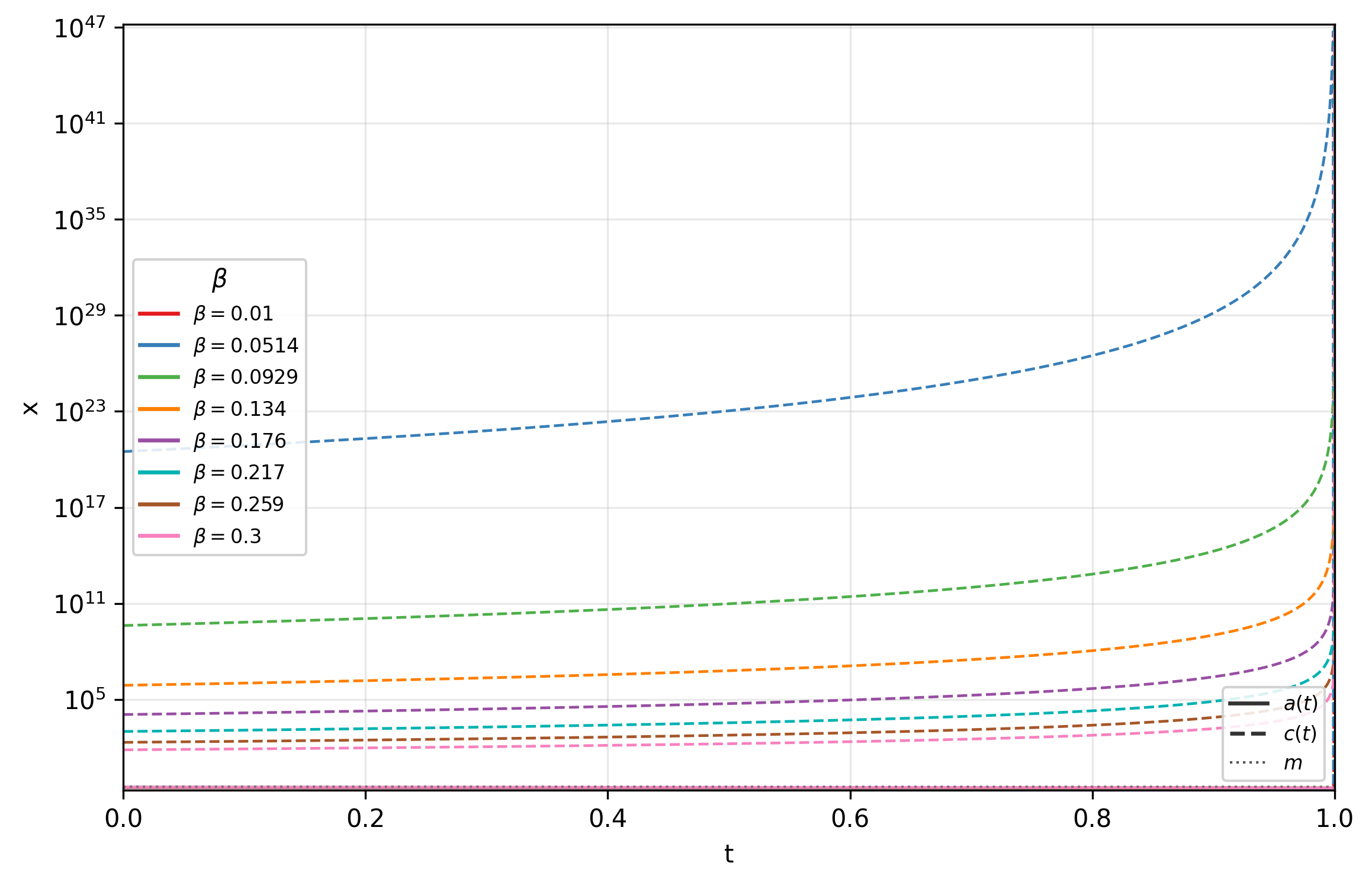}};
        \begin{scope}[x={(main.south east)}, y={(main.north west)}]
            \coordinate (srcSW) at (0.115,0.055);
            \coordinate (srcNE) at (0.965,0.150);
            \node[srcbox, fit=(srcSW)(srcNE), inner sep=0pt] (src) {};
            \node[zoombox, anchor=east] (zoom) at (-0.05,0.46)
                {\includegraphics[width=0.31\textwidth]{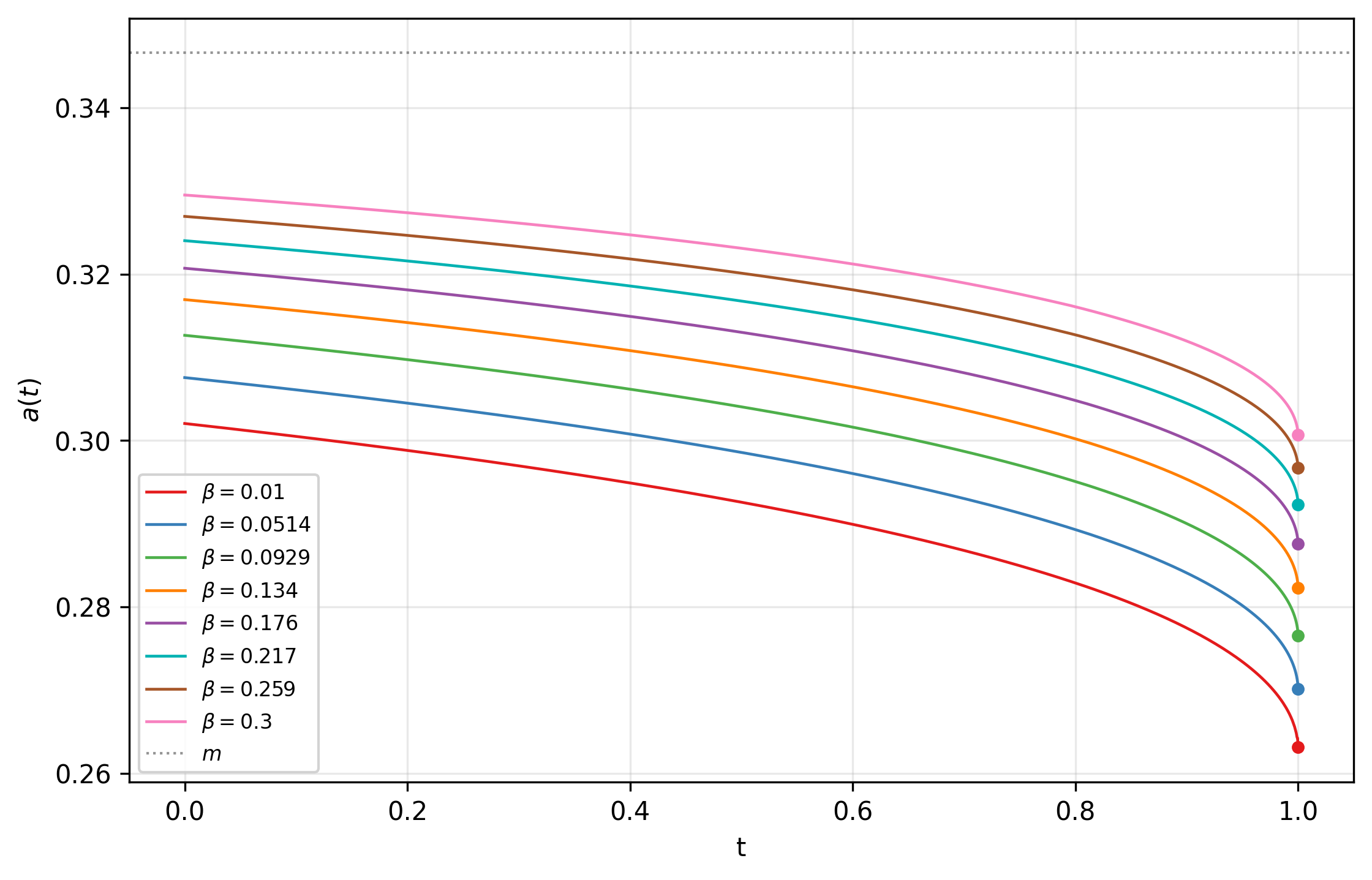}};
            \draw[conn] (src.north west) -- (zoom.north east);
            \draw[conn] (src.south west) -- (zoom.south east);
        \end{scope}
    \end{tikzpicture}
    \caption{The resistance boundaries $a(t)$ and $c(t)$ for eight equally spaced elasticity coefficients $\beta\in[0.01,0.30]$, with $c(t)$ shown on a logarithmic scale. The framed panel enlarges the boxed strip, magnifying the lower resistance boundaries $a(t)$. The final rise of each $c$-curve out of the top frame at $t=T$ depicts $c(T-)=+\infty$ and is not a computed value; every plotted point for $t<T$ is the raw output of the recursion.}
    \label{fig:positive-comparative}
\end{figure}

In Figure 4, for every given and fixed $t$, we see that the support boundary $b(t)$ increases as $\beta$ decreases. Since in the negative drift case the stopping region is given by $x\leq b(t)$, this means that the stopping area also increases as $\beta$ decreases. Financially speaking, the support boundary lies in the low-price region, where the price process is less than 1. Hence, when $x<1$, a smaller value of $\beta$ leads to a higher relative volatility. Under the negative drift $\mu<0$, this makes waiting more risky, since the price may move further below the representative aspiration level. Therefore, the trader is more willing to stop at a relatively higher price, and the support boundary $b(t)$ is pushed upward.

On the contrary, when $\beta$ increases, the volatility in the low-price region becomes smaller. In this case, the trader has more confidence to wait for a price that is closer to the aspiration level, since the risk of a large downward movement is relatively lower. This explains why, for each fixed $t$, the boundary $b(t)$ decreases as $\beta$ increases.

\begin{figure}[htbp]
    \centering
    \IfFileExists{mu_negative_comparative_boundary.png}{%
    \includegraphics[
        width=0.65\textwidth,
        height=0.25\textheight,
        keepaspectratio
    ]{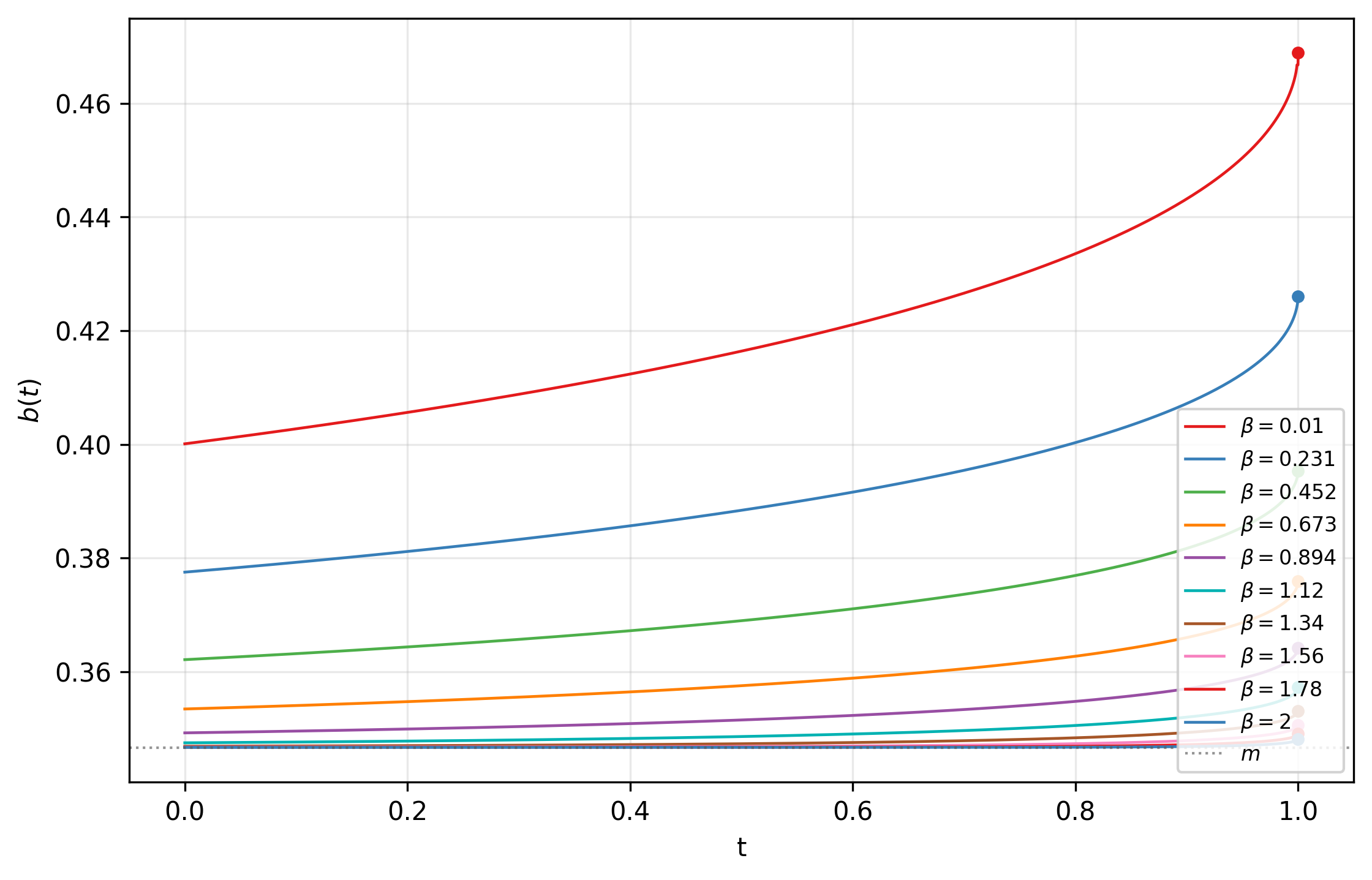}}{%
    \fbox{\parbox[c][0.25\textheight][c]{0.65\textwidth}{\centering \texttt{mu\_negative\_comparative\_boundary.png}}}}
    \caption{The support boundaries $b(t)$ for ten equally spaced elasticity coefficients $\beta\in[0.01,2]$, computed with $\mu=-0.15$, $\sigma=0.30$, $\lambda=2$ and $T=1$ as in Example~\ref{support}. For each fixed $t$, the boundary $b(t)$ increases as $\beta$ decreases.}
    \label{fig:negative-comparative}
\end{figure}

\section{Conclusion}

In this paper we extended the optimal prediction approach to resistance and support levels from the geometric Brownian motion setting to the CEV setting. By modelling the aspiration level as a hidden random variable independent of the price process, we reduced the original prediction problem to a finite-horizon optimal stopping problem with a convex loss function. Under suitable admissibility assumptions on the distribution of the aspiration level, we proved the existence and uniqueness of the corresponding optimal boundaries and derived nonlinear integral equations that characterise them.

The CEV specification leads to a different structure of the stopping region from the geometric Brownian motion case. In the positive drift case, the stopping region is described by a band bounded by two curves \(a(t)\) and \(c(t)\), while in the negative drift case it is described by a single support boundary \(b(t)\). The numerical examples illustrate these theoretical results and show how the integral equations can be solved by backward induction.

The comparative statics further demonstrate the effect of the elasticity coefficient on the optimal boundaries. In particular, decreasing \(\beta\) enlarges the stopping region in both the positive and negative drift cases, and the numerical results provide evidence that the upper boundary \(c(t)\) tends to infinity as \(\beta\to0\), which is consistent with the geometric Brownian motion case. These observations show that the level-dependent volatility in the CEV model plays an important role in the optimal prediction of resistance and support levels.

\section*{Acknowledgements}

I would like to thank my PhD supervisor, Professor Goran Peskir, for his continuous guidance, valuable discussions, and many helpful suggestions during the preparation of this paper.

\bibliographystyle{spmpsci}
\bibliography{ref}

\end{document}